\documentclass[fleqn,11pt,letterpaper]{article}

\usepackage{amsfonts}
\usepackage{amssymb}
\usepackage{amsmath}
\usepackage{amsthm}
\usepackage{enumitem}
\usepackage{tabularx}
\usepackage{multirow}
\usepackage{float}
\usepackage{graphicx}
\usepackage{epstopdf}
\usepackage{color}
\usepackage[margin=2.57cm]{geometry}
\definecolor{darkred}{rgb}{0.5,0.2,0.2}
\usepackage[pdftitle={Roughness_signature_functions},pdfauthor={Peter Christensen},colorlinks=TRUE,allcolors=darkred]{hyperref}
\usepackage{booktabs}
\usepackage{pdflscape}
\usepackage[round]{natbib}
\usepackage{subcaption}
\usepackage{mwe}
\usepackage{bbm}
\usepackage{float}
\floatstyle{plaintop}
\restylefloat{table}

\restylefloat{table}

\theoremstyle{plain}
\newtheorem{theorem}{Theorem}[section]

\newtheorem{assumption}{Assumption}[section]

\newtheorem{corollary}{Corollary}[section]

\newtheorem{lemma}{Lemma}[section]
\newtheorem{proposition}{Proposition}[section]

\theoremstyle{definition}

\theoremstyle{remark}
\newtheorem{remark}{Remark}[section]

\def\bi{\begin{itemize}}
\def\ei{\end{itemize}}

\allowdisplaybreaks
\graphicspath{{Figures/}}
\numberwithin{equation}{section}

\newif\ifi

\usepackage{caption}
\begin{document}

\title{Roughness Signature Functions}

\author{
Peter Christensen\thanks{
Department of Economics and Business Economics, 
Aarhus University, 
Fuglesangs All\'e 4,
8210 Aarhus V, Denmark.
E-mail:\
\href{mailto:pchr@econ.au.dk}{\nolinkurl{pchr@econ.au.dk}}. 
}
}

\maketitle
\begin{abstract}
    Inspired by the activity signature introduced by \cite{Todorov2010}, which was used to measure the activity of a semimartingale, this paper introduces the roughness signature function. The paper illustrates how it can be used to determine whether a discretely observed process is generated by a continuous process that is rougher than a Brownian motion, a pure-jump process, or a combination of the two. Further, if a continuous rough process is present, the function gives an estimate of the roughness index. This is done through an extensive simulation study, where we find that the roughness signature function works as expected on rough processes. We further derive some asymptotic properties of this new signature function. The function is applied empirically to three different volatility measures for the S\&P500 index. The three measures are realized volatility, the VIX, and the option-extracted volatility  estimator of \cite{Todorov2019}. The realized volatility and option-extracted volatility show signs of roughness, with the option-extracted volatility appearing smoother than the realized volatility, while the VIX appears to be driven by a continuous martingale with jumps. 

\end{abstract}

\bigskip

\noindent {\bf Keywords}:  Rough volatility, Lévy process, Brownian semi-stationary process, Activity index, Blumenthal-Getoor index, Continuous time. \\

\baselineskip=17pt

\newpage

\section{Introduction}\label{sec:introduction}

In this paper, we introduce the roughness signature function, which is a similar statistic to the activity signature function introduced by \cite{Todorov2010}, and explore its behaviour when used on continuous time processes with sample paths that are less regular than those of a Brownian motion. The activity signature function was initially introduced to measure the activity index of an Itô semimartingale. The activity index is a generalization of the Blumenthal-Getoor index of pure-jump Lévy processes to general semimartingales, and it thus describes the fine structure of the process. The higher the activity index, the more vibrant the process will be, where a process of finite activity will have an activity index of $0$, while a continuous martingale will have an activity index of $2$. The estimator of the activity index is based on a ratio of the $p$'th realized power variation of the process, sampled at two different frequencies, and evaluating the estimator over different values of $p$ one can plot it as a function, which is what is termed the activity signature function. The shape and level of the curve will reveal if a process contains a continuous Brownian-driven component, a pure jump component, or both, and the activity of the jumps if they are present. 

For a general purely continuous process, the fine structure can be described various ways, including the Hôlder continuity and the fractal dimension. A common way to characterize the fine structure, is through the limiting behaviour of the $p$-variation. This is for example what is done in \cite{HanSchied2021}. Following them, we will say that a continuous function $x:[0,T]\to\mathbb{R}$ has roughness exponent $H$ if, for any partition $\{t_0,...,t_n\}$ of $[0,T]$ with vanishing grid size, it is the case that
\begin{equation}\label{eq:roughness_def}
    \lim_{n\to\infty}\sum_{i=1}^n\vert x(t_i) - x(t_{i-1})\vert^p= \begin{cases}0 & \text { for } p>1 / H, \\ \infty & \text { for } p<1 / H .\end{cases}
\end{equation}
This roughness index will be informative about how smooth the paths of the function is, where the higher the roughness index, the smoother the function. For a fractional Brownian motion (fBm) with Hurst exponent $H$, which is the cannonical example of the rough process, the paths will exactly have roughness index equal to the Hurst exponent. For a general Brownian semi-stationary ($\mathcal{BSS}$) process, \cite{Corcuera2013} introduces an estimator of the roughness index of a process, which, like the roughness signature function, is based on a ratio of power variations at different frequencies. Motivated by the similarity of the roughness index and the activity index, as well as the similarity of the \cite{Todorov2010} and \cite{Corcuera2013} estimators, this paper introduces the roughness signature function, and illustrates how it can be used to distinguish between pure-jump processes and processes that are rougher than a Brownian motion, and that it can recover the Hölder continuity of the process when a continuous part is present. In particular, we consider the case where the Brownian motion is replaced by a $\mathcal{BSS}$ process, which has the flexibility to exhibit a roughness index $H$ for any $H\in(0,1)$. We find that we can distinguish between the different processes and that when the underlying process is a $\mathcal{BSS}$ process, we can recover the roughness index. 

Being able to differentiate between rough processes and processes driven by jumps is relevant in many settings. Recently, a new generation of stochastic volatility models termed rough stochastic volatility models have emerged, in which the volatility is modelled using processes that are rougher than a Brownian motion. They were initially introduced by \cite{Gatheral2018}, who found rough volatility models to be remarkably consistent with empirical observations, including the power-law at-the-money skew of implied volatility. There is, however, also support for volatility being a discontinuous jump process, both by the very successful stochastic volatility models that incorporate jumps such as the jump-diffusion models of \cite{Duffie2000}, as well as the non-parametric results by \cite{Todorov2011} implying that volatility should be modelled as a high-activity pure-jump process. Being able to tell whether volatility jumps or is rough, and knowing the activity level of the jumps or Hurst exponent of the continuous part, would help practitioners within the financial sector use more realistic models. 

The approach also has applications outside of finance. An important topic within statistical physics is turbulence modelling. For turbulence modelling, both stable Lévy processes and rough processes have been considered (see e.g. \cite{Takayasu1984} and \cite{Corcuera2013}). Within this topic, it is also interesting to be able to discern which model is correct. It is also interesting to recover the smoothness index of the process if the process is found to be rough. The reason is that a smoothness-estimate of turbulence of $-1/6$ in what is called the inertial range will exactly correspond to Kolmogorov's celebrated $5/3$ law \citep{Kolmogorov1941a,Kolmogorov1941b}, which has also been studied using statistics similar to the ones in this paper in  \cite{Bennedsen2020} and \cite{Corcuera2013}.

This paper is closely related to a number of existing papers. As stated above, the idea behind the roughness signature function is to consider a ratio of power variation of the observed process sampled at two different frequencies and then utilize the fact that the speed of convergence of the realized power variation only depends on the activity index. A similar idea has been utilized in many other papers within the high-frequency econometrics literature.  In  \cite{Corcuera2013}, an estimator for the smoothness parameter of a $\mathcal{BSS}$ process based on the realized power variation is introduced. This structure of this estimator is very similar to the \cite{Todorov2010} estimator, as will be shown below. The statistic considered in \cite{Corcuera2013} was based on second-order increments rather than first-order increments, which was done to ensure that their results are valid for all smoothness levels, but as we shall see below, when our interest is restricted to rough processes, the same results carry through with first-order increments. 

Using realized power variation to study the fine properties of a process for which we have discrete high-frequency observations is not unique to the two articles on which we base our analysis. \cite{Woener2011} considers an identical problem to the one considered in this paper, i.e., discriminating between a process driven by a rough process and a pure-jump process and estimating the roughness or activity using realized power variation. However, they use a statistic based on the log-power variation rather than a ratio of power variations. The resulting statistic yields estimates of the roughness or activity that are quite biased, limiting its usefulness in practice. Further, they restrict their rough processes to be driven by fractional Brownian motions, whereas we allow for the more flexible class of $\mathcal{BSS}$ processes.
\cite{Bennedsen2020} considers the same problem as \cite{Corcuera2013} but focuses on lowering the estimator's variance by using more frequencies and considers how to handle the microstructure noise present in high-frequency financial observations. In concurrent work with the present paper, \cite{ChongTodorov2023} develops a formal test for the roughness of volatility, which is based on looking at the autocovariance of a measure of spot volatility. In \cite{Ait-Sahalia2009}, a similar statistic is introduced to test whether jumps are present in the process. This statistic will converge to 1 if jumps are present. The setting here is very different from the one in this paper, as they assume a continuous part is present and want to test if there is additionally a jump component.  In \cite{Ait-Sahalia2009a}, a similar problem is considered, as they try to estimate the activity level of the jumps of a process when a continuous part is present. One of their statistics is also based on a change-of-frequency approach, but here all small increments are discarded to isolate the jump part, and the statistic is therefore very different. 

The project is also related to the rapidly growing literature on rough stochastic volatility models. This strain of literature was started by the seminal paper by \cite{Gatheral2018}, who were the first to document the fit of rough volatility models to empirical observations. As mentioned, they find that modelling log-volatility using a rough process makes it possible to match the power-law structure of the at-the-money implied volatility surface, as well as the monofractal scaling relationship of the $p$-variation of log-volatility. Following this seminal article, the notion that volatility should be rough has seen large support. Many papers have confirmed that volatility series behave like rough processes using more sophisticated estimators and different volatility indices. These include \cite{Bennedsen2021}, who considers a broader class of assets and more complex rough processes,  \cite{BCSV2020}, who develops a GMM estimator, which takes the error when estimating the latent volatility variable into account, and \cite{LivieriEtAl2018}, who considers estimating the roughness of volatility using the Black-Scholes implied volatility of short-dated options rather than the high-frequency return-based volatility estimators of the other papers. Further, \cite{WangEtAl2023} has studied a method-of-moments-based estimator for the fractional Ornstein-Uhlenbeckm model and applied to to modelling realized variance, while \cite{BennedsenEtAl2023} has developed a composite likelihood estimator for Gaussian moving average processes and used it to estimate the roughness of the spot volatility of bitcoin prices. Several papers further consider how to specify rough volatility models that can be used in practice. Two models that have gained significant traction are the rough Bergomi model of \cite{Bayer2016a} and the rough Heston model of \cite{ElEuch2018}. In both these models, volatility is driven by variations of the fractional Brownian motion.

The structure of this paper is as follows. Section \ref{sec:activity} describes our setup, introduces the necessary assumptions and formally defines the index of activity as well as the roughness index.  Section \ref{sec:ASF} formally defines the roughness signature function and presents the asymptotic results for the roughness signature function. Section \ref{sec:sim_study} carries out a simulation study, documenting the performance of the roughness signature function in finite samples. Section \ref{sec:emp_application} applies the roughness signature function in an empirical application, studying which processes can describe various volatility measures. Section \ref{sec:conclusion} concludes. 

\section{The Index of Activity}\label{sec:activity}

\subsection{Setup}\label{subsec:setup}
In this section, we will introduce the different processes we are working with and the different indices we will be working with. The main goal is to measure the ``activity'' of a given stochastic process, $\mathcal{X}$. Much like \cite{Todorov2010}, we will assume that there are three potential underlying processes for generating our observations, which are defined on some probability space $(\Omega,\mathcal{F},\mathbb{P})$ equipped with a filtration $\mathbb{F}$ satisfying the usual conditions, such that the processes are adapted to this filtration. 

The continuous model: 
\begin{equation} \label{eq:cont_model}
    X_{t}=\int_{-\infty}^{t}g(t-s)\sigma_{1s}dW_s
\end{equation}

The pure-jump model: 
\begin{equation} \label{eq:jump_model}
    Y_{t}=\int_{0}^{t}b_{2s}ds+\int_{0}^{t}\int_{\mathbb{R}}\sigma_{2s-}\kappa(x)\tilde{\mu}(ds,dx)+\int_{0}^{t}\int_{\mathbb{R}}\sigma_{2s-}\kappa'(x)\mu(ds,dx)
\end{equation}

and the continuous plus jumps model: 
\begin{equation} \label{eq:cont_jump_model}
    Z_{t}=X_{t}+Y_{t}
\end{equation}

Here, $W=(W_t)_{t\in\mathbb{R}}$ is a two-sided Brownian motion, $g:\mathbb{R}^+\to\mathbb{R}$ is a deterministic weight-function such that $g\in L^2(\mathbb{R}^+)$, and $\sigma_1=(\sigma_{1t})_{t\in\mathbb{R}}$ and $\sigma_2=(\sigma_{2t})_{t\in\mathbb{R}^+}$ are cádlág processes. $\mu$ is a jump measure on $\mathbb{R}_{+}\times\mathbb{R}$ with compensator $\nu(x)dxds$, $\tilde{\mu}(ds,dx):=\mu(ds,dx)-\nu(x)dxds, \kappa'(x)=x-\kappa(x)$ where $\kappa(x)$ is a continuous truncation function. If the process $X$ contains a jump-part, we will always assume the observed path contains at least one jump. 

\begin{remark}
    We note that the classes of processes considered are quite broad, and contain many interesting examples known from financial econometrics. For example, \eqref{eq:cont_model} encompasses all Gassuian models, as long as they satisfy the so-called pure-nondeterminism assumption, which includes the fractional Ornstein-Uhlenbeck model studied in many papers---see e.g. \cite{BennedsenEtAl2023}. The $\sigma$ then allows one to add stochastic volatility (of volatility) to these. \eqref{eq:jump_model} will included many of the well-known Lévy processes seen in finance, for example the famous CGMY model of \cite{CarrEtAl2002}.
\end{remark}


The main difference between this setup and the one considered in \cite{Todorov2010} is the specification of the continuous process. \cite{Todorov2010} only considers Itô semimartingale processes. In this setup, we allow the process to be a Brownian semi-stationary ($\mathcal{BSS}$) process, a class of processes introduced by \cite{ole-jurgen09}. This class of processes can be semi-martingales under certain conditions, but this will not generally be true. In particular, this class of processes is able to have sample paths are rougher than those  of a Brownian motion, almost surely, and have thus been proposed as models for volatility \citep{Bennedsen2021}. 

\subsection{Assumptions}\label{subsec:assumptions}

We will need to make some assumptions on the processes introduced above to get existence of the processes and the convergence results. Most of our assumptions will be shared with \cite{Todorov2010}, but some modifications have to be made, seeing as we have left the semimartingale setting. To introduce the assumptions of this paper, we need some further notation. First, we consider the centered stationary Gaussian process $G=(G_t)_{t\in\mathbb{R}}$, called the Gaussian core of $X$, which for each $t$ is defined as 
\begin{equation}
    G_t:=\int_{-\infty}^tg(t-s)dW_s
\end{equation}
that is, $G$ is defined as $X$, but where the stochastic volatility $\sigma_{1}$ is set to constantly equal $1$. Note that this process is well-defined as $g$ is square integrable. We define the correlation kernel $r$ of $G$ as 
\begin{equation}
    r(t):=\frac{\int_0^\infty g(u)g(u+t)du}{\Vert g\Vert^2_{L^2(\mathbb{R}^+)}}
\end{equation}
and the variogram, $R$, of the kernel as 
\begin{equation}
    R(t):=\mathbb{E}[(G_{t+s}-G_s)^2]=2\Vert g\Vert^2_{L^2(\mathbb{R}^+)}(1-r(t))
\end{equation}

First, we have the following regularity conditions on the kernel function $g$
\begin{assumption}\label{assumption:kernel}
It holds that
\begin{enumerate}
    \item $g(x)=x^\alpha L_g(x)$.
    \item $g^{(k)}(x)=x^{\alpha-k}L_{g^{(k)}}(x)$ and, for any $\varepsilon>0$, we have $g^{(k)}\in L^2((\varepsilon,\infty))$. Furthermore, $\vert g^{(k)}\vert$ is non-increasing on $(a,\infty)$ for some $a>0$.
    \item For any $t>0$
    \begin{equation}
        F_t=\int_1^\infty\vert g^{(k)}(s)\vert^2\sigma_{t-s}^2ds<\infty.
    \end{equation}
\end{enumerate}
here, for each smooth function $f$, $f^{(k)}$ denotes the $k$'th derivative of the function, while $L_f$ denotes a continuous function that is slowly varying around zero, and $\alpha\in(-1/2,1/2)\setminus\{0\}$.
\end{assumption}

Secondly, we need some assumptions on the variogram $R$ of the Gaussian kernel
\begin{assumption}\label{assumption:variogram}
For the smoothness parameter $\alpha$ from Assumption \ref{assumption:kernel} it holds that
\begin{enumerate}
    \item $R(x)=x^{2\alpha+1}L_R(x)$.
    \item $R^{(2k)}(x)=x^{2\alpha-2k+1}L_{R^{(2k)}}(x)$.
    \item There exists a $b\in(0,1)$ such that 
    \begin{equation}
        \limsup_{x\downarrow0}\sup_{y\in[x,x^b]}\left\vert\frac{L_{R^{(2k)}}(y)}{L_R(x)}\right\vert<\infty
    \end{equation}
\end{enumerate}
\end{assumption}


The assumptions relating to the kernel function is to ensure that the sample paths of the continuous process exhibits roughness, and the index $\alpha$ will exactly correspond to the roughness index, which will be defined below. The assumptions on the Gaussian core are required for deriving the asymptotic theory for power variation of differences of the Gaussian process $G$. In particular, Assumption \ref{assumption:variogram} implies that the small scale increments of $G$ behaves similar to those of a fractional Brownian motion $B^H$ with Hurst index $H=\alpha+1/2$. This implies that the limit theory of the power variation of $G$ is the same as that of the power variation of $B^H$. This relation is formalized in \cite{Corcuera2013}. These assumptions are standard in the literature when working with $\mathcal{BSS}$ processes, see e.g. \cite{Corcuera2013} and \cite{ole_jose_mark13}.

Next, we need some assumptions on the jump-process. First, we have some mild regularity conditions on the drift, volatility, and intensity processes
\begin{assumption}\label{assumption:levy_regularity}
The processes $\sigma_{2}$ and $b_{2}$ have cádlág paths; both of the processes $\sigma_{2}$ and $\sigma_{2-}=(\sigma_{2t-})_{t\in[0,\infty)}$  are everywhere different from zero. 
\end{assumption}

Secondly, we will an assumption on the behavior of the Lévy measure around zero. 
\begin{assumption}\label{assumption:levy_measure}
The Lévy density $\nu(x)$ can be decomposed as 
\begin{equation}
    \nu(x)=\nu_1(x) + \nu_2(x)
\end{equation}
\begin{equation}
    \nu_1(x) = \frac{A}{x^{\beta+1}} \mathbbm{1}_{\{x>0\}} + \frac{B}{\vert x\vert^{\beta+1}} \mathbbm{1}_{\{x<0\}}
\end{equation}
\begin{equation}
    \vert\nu_2(x)\vert\leq\frac{\phi(x)}{\vert x\vert^{\beta'+1}}
\end{equation}
where $A$ and $B$ are non-negative constants with $A+B>0$; $0\leq \beta' < \beta < 2$; and $\phi(x)$ is some non-negative slowly varying around zero function which is bounded at zero. 
\end{assumption}
Then, we will make some fairly technical assumptions on the relation between the drift term and the Lévy measure. 
\begin{assumption}\label{assumption:levy_drift}
    \begin{enumerate}
        \item If $\beta\leq1$ we assume that
    \begin{equation}
        \int_0^tb_{2s}ds-\int_0^t\int_\mathbb{R}\sigma_{2s}\kappa(x)\nu(x)dxds\equiv0,\quad\text{for every }t>0.
    \end{equation} 
    \item If $\beta=1$, then $\nu(x)$ and $\kappa(x)$ are symmetric and $b_{2t}\equiv0$ for every $t>0$. 
    \end{enumerate}
\end{assumption}

The assumption \ref{assumption:levy_measure} concerns the behaviour of the jump measure around zero and is less trivial. We are essentially imposing that the Lévy measure should behave like that of a stable process around zero. This somewhat strict assumption is necessary to get the convergence of the realized variation we are interested in when considering a pure-jump process and powers less than the Blumenthal-Getoor index of the process. When we are imposing this assumption, the Blumenthal-Getoor index of the pure-jump process will exactly coincide with $\beta$. The intuition is that for a $\beta$-stable process, the Blumenthal-Getoor index will be $\beta$, and for any pure-jump process, the Blumenthal-Getoor index will purely be determined by the small jumps. This is because there is always a finite number of jumps greater than any fixed threshold. When the Lévy measure behaves like that of a $\beta$-stable process for small values, we will thus get the same index. This assumption is very typical when working with estimators based on infill asymptotics of realized power variation statistics, see e.g. \cite{Todorov2010}, \cite{Ait-Sahalia2009a} and \cite{Woener2011}.

The third assumption essentially states that that when we have finite activity jumps, there should be no drift. This is done to ensure that the drift term $b_{2t}$ and the compensation of the jumps do not overpower the activity of the jump measure. This is to ensure that we can extract the activity index of the jumps in the pure-jump case. We note that this can be a quite strict assumption in financial econometrics, where many processes are known to exhibit jumps. If one is worried about a drift, a version of the roughness signature function based on second order increments can be constructed, which will be robust to the presence of a drift for low-activity jump processes---this follows from the results in \cite{Todorov2013}. 

\subsection{The Index of Activity}

We assume that we are observing a stochastic process, $\mathcal{X}$, on a fixed time interval $[0,T]$. We will assume that we observe it at equidistant times, $0,\Delta_n,2\Delta_n,...,\lfloor\Delta_n/t\rfloor\Delta_n$, with distance $\Delta_n$ between the observations. We will be considering infill asymptotics, i.e. we will think of $\Delta_n$ as being small and study the behaviour of our estimators as $\Delta_n\to0$. 

To introduce the activity index, we will first define the realized $p$-variation of a path of $\mathcal{X}$ sampled at frequency $\nu\Delta_n$: 
\begin{equation} \label{eq:p_variation_def}
    V(\mathcal{X},p,\nu,\Delta_n)_t=\sum_{i=\nu}^{\lfloor\Delta_n/t\rfloor}\vert\Delta_i^{n,\nu}\mathcal{X}\vert^p,\quad p>0,t\in[0,T]
\end{equation}

where $\Delta_i^{n,\nu}\mathcal{X}:=\mathcal{X}_{i\Delta_n}-\mathcal{X}_{(i-\nu)\Delta_n}$. Estimators based on this statistic is widely used. \cite{Todorov2010} covers convergence results of this statistic for semimartingales, while \cite{Corcuera2006} and \cite{Corcuera2013} covers it for integrals with respect to fractional Brownian motion and general $\mathcal{BSS}$ processes, respectively. 

Having defined this, we follow \cite{TodorovTauchen2011} and define the index of activity of an observed path of $\mathcal{X}$ as 
\begin{equation}
    \beta_{\mathcal{X},T} = \inf\left\{r>0:\underset{\Delta_n\to0}{\operatorname{plim}}V(X,r,1,\Delta_n)_T<\infty\right\}
\end{equation}

This quantity is defined pathwise, but under our assumptions, it will be the same on all paths, up to a set of paths with measure $0$. In the semimartingale setting of \cite{Todorov2010}, this index will always lie in the interval $[0,2]$, as the quadratic variation of a semimartingale always exists. In the more general setting we are in, we will see that it can take arbitrarily large values. For a pure-jump process, this activity index will exactly recover the (generalized) Blumenthal-Getoor index, as it reduces to the generalization of the Blumenthal-Getoor index of \cite{Ait-Sahalia2009a}, while for a continuous martingale, it will equal $2$. If a process is made up of multiple parts, e.g. both a jump-part and a continuous martingale, the activity index of the combined process will be given by the activity of the most active component.

For a general continuous process, the activity index will reflect the roughness of the sample paths of the process. It coincides with the reciprocal \eqref{eq:roughness_def}, which \cite{HanSchied2021} has coined the Hurst roughness exponent. In general, this index will be informative about the regularity of the sample paths of the process. The higher this activity index is - corresponding to a lower roughness index - the less regular the sample paths will be. For a fractional Brownian motion, this Hurst roughness index will exactly coincide with the classical Hurst  index, see e.g. Proposition 2.1 of \cite{Bennedsen2020}. For the more general class of processes we consider, this index will be directly related to the parameter $\alpha$ of Assumption \ref{assumption:kernel}, as it will always be the case that 
\begin{equation}
    \beta_{X,T}=\frac{1}{\alpha+1/2}
\end{equation}
and it is therefore natural to think of $\alpha$ as a ``roughness index''. 

For both the continuous and the pure-jump specification of the underlying process, the index of activity is thus informative about the ``fine structure'' of the sample paths of the process, and an estimator of the activity index will thus help us understand which type of model is appropriate to describe an observed time series. In this respect, it is worth noting that a Brownian motion -- the most ``active'' semi-martingale -- has activity index $2$, and continuous processes that are ``rougher'' than a Brownian motion will thus have an an activity index that is higher than $2$. The index of activity can thus separate the pure-jump and rough purely-continuous processes in our setup, as all the pure-jump processes will have activity less than $2$, while all rough purely continuous processes will have an index larger than $2$. 

\section{Roughness Signature Function}\label{sec:ASF}
\subsection{Definition}
For estimating the activity (and hence roughness) index, we will be relying on how the scaling used to achieve convergence depends on the activity index. In particular, the difference in the scaling needed for convergence of the power variation of the discontinuous processes when the power is below and above the index of activity. To be more concrete, for all the processes we consider it will be the case that for powers below the index of activity
\begin{equation}
    k^{1-p/\beta_{\mathcal{X},T}}\Delta_n^{1-p/\beta_{\mathcal{X},T}}V(\mathcal{X},p,1,k\Delta_n)_T\overset{\mathbb{P}}{\longrightarrow}\Xi_T(p)
\end{equation}
as $\Delta_n\to0$, where $\Xi_T(p)$ is a stochastic process depending on the underlying process driving the observations of $\mathcal{X}$. This suggests that for low powers, we can recover the activity level of the paths by a regression of $\log(V(\mathcal{X},p,1,k\Delta_n))$ on $\log(k\Delta_n)$. This is exactly how the activity signature function is derived in \cite{Todorov2010}. If only two points are considered, we get that the estimate of the activity will exactly be 

\begin{equation}\label{eq:def_asf_TT}
    \hat{\beta}^{\text{T\&T}}(\mathcal{X},p,k,\Delta_n)_T = \frac{\log(k)p}{\log(k)+\log(V(X,p,k\Delta_n)_T)-\log(V(X,p,\Delta_n)_T)}
\end{equation}

If we consider this estimator as a function of $p$, we arrive at the original activity signature function. The idea is now to examine its behaviour for different values of $p$. Based on the arguments above, it is intuitively clear that for powers lower than the activity level $\beta_{\mathcal{X},T}$, the activity signature function will converge to $\beta_{\mathcal{X}X,T}$. For power powers higher than $\beta_{\mathcal{X},T}$, the convergence will depend on whether a jump-part is present or not. If it is, no scaling factor is needed for the convergence of $V(\mathcal{X},p,k\Delta_n)$, and we would thus expect the activity signature function to converge to $p$, as we would have that $\log(V(\mathcal{X},p,k\Delta_n))-\log(V(\mathcal{X},p,\Delta_n))$ would converge to zero. If no jump part is present, the scaling factor would still depend on the sampling frequency, and we would expect it to remain at the activity index. These intuitive results are formalized for semimartingales in \cite{Todorov2010}.

We will base our roughness signature function on the same idea, but we will make two alterations. First, we define sampling at lower frequencies differently than what is done in \cite{Todorov2010}. By their definition, sampling at half frequency effectively amounts to discarding every other observation, i.e. the sample is halved. We will instead be using the definition of \cite{Corcuera2013}, and define sampling at half frequency by taking overlapping increments, and hence arriving at a sample size of $n-1$ when sampling at half frequency. Intuitively this should lead to a more efficient estimator, which was also verified in our simulation experiments. This corresponds to controlling sampling speed by using the $\nu$ parameter in \eqref{eq:p_variation_def}, instead of multiplying onto the $\Delta$, as was done in equation \eqref{eq:def_asf_TT}. Secondly, our main motivation for studying the roughness signature function is to examine the roughness of volatility, and from the recent literature on the topic, we expect volatility to be very rough. In their pioneering article, \cite{Gatheral2018} found evidence that volatility across a wide range of indices behaved as a fBm with Hurst exponent $H\in(0.07,0.15)$, i.e. a very rough process. In subsequent articles, it has been shown to be even rougher, with e.g. \cite{BCSV2020} finding evidence of $H$ as low as $0.02$. In light of this, using the regular activity signature function may lead to issues, as we have seen that the level of the activity signature function is the inverse of the Hurst exponent, which for a very low Hurst exponent will give a very high and unstable signature function. We will therefore be considering the reciprocal of the the activity signature function.  

We therefore define the \emph{roughness signature function} as follows

\begin{equation}\label{eq:def_asf}
    \hat{H}(\mathcal{X},p,\nu,\Delta_n)_T = \frac{\log(V(X,p,\nu,\Delta_n)_T)-\log(V(X,p,1,\Delta_n)_T)}{\log(k)p}
\end{equation}

\begin{remark}  
    This roughness signature function is essentially equivalent to the COF-estimator of roughness introduced in \cite{Corcuera2013}. The main difference is that they chose to base their estimator on second order increments to get asymptotic normality for the the full spectrum $\alpha\in(-1/2,1/2)$, while we base the roughness signature function on first order increments, which are sufficient for rough processes. Further, their estimator is only developed for point estimates, i.e. for a fixed $p$, typically $p=2$.  We could have also defined the roughness signature function based on second order increments, which would work in an identical way to the one defined above, but this would lead to a less efficient statistic, due to the loss of information when differencing the second time. This is formally derived for pure-jump semimartingales in \cite{Todorov2013}. The ``gain'' from using second order increments would however be increased robustness to the pressence of a drift, in particular when estimating the activity of finite variation pure-jump processes, as mentioned earlier. 
\end{remark}

When working with the roughness signature function, we will typically use $\nu=2$, i.e. increments sampled at half the speed of the original series, and will denote $\hat{H}(\mathcal{X},p,\Delta_n)_T:=\hat{H}(\mathcal{X},p,2,\Delta_n)_T$. It is this variation of the roughness signature function that we will focus on in the following section.

\subsection{Asymptotic Results}\label{sec:asymptotics}
In our setting, we get the following convergence of the roughness signature function: 

\begin{theorem}\label{thm:consistency}
 Fix $T>0$, and let $H=(\alpha+1/2)^{-1}$, where $\alpha$ is the smoothness index of Assumption \ref{assumption:kernel}, and $\beta_{\mathcal{X},T}$ denote the index $\beta$ of Assumption \ref{assumption:levy_measure}. Then we have, for $\Delta_n\to0$, 
\begin{enumerate}
    \item If $\mathcal{X}\equiv X$, and assumptions \ref{assumption:kernel} and \ref{assumption:variogram} are satisfied, then we have 
    \begin{equation}
        \hat{H}(\mathcal{X},p,\Delta_n)_T\overset{\mathbb{P}}{\longrightarrow}H
    \end{equation}
    where the convergence is locally uniform in $p$ on $(0,\infty)$
    \item If $\mathcal{X}\equiv Y$, and assumptions \ref{assumption:levy_regularity}, \ref{assumption:levy_measure} and \ref{assumption:levy_drift} are satisfied, then
    \begin{equation}
        \hat{H}(\mathcal{X},p,\Delta_n)_T\overset{\mathbb{P}}{\longrightarrow}
        \begin{cases}
            1/\beta_{\mathcal{X},T} & \text{if }p<\beta_{\mathcal{X},T}\\
            1/p & \text{if }p>\beta_{\mathcal{X},T}
        \end{cases}
    \end{equation}
    where the convergence is locally uniform in $p$ on $(0,\beta_{X\mathcal{X},T})\cup(\beta_{\mathcal{X},T},\infty)$
    \item If $\mathcal{X}\equiv Z$ with atleast one jump on the path, and Assumption \ref{assumption:kernel} and \ref{assumption:variogram} are satisfied, then we have 
    \begin{equation}
        \hat{H}(\mathcal{X},p,\Delta_n)_T\overset{\mathbb{P}}{\longrightarrow}
        \begin{cases}
            H & \text{if }p<H^{-1}\\
            1/p & \text{if }p>H^{-1}
        \end{cases}
    \end{equation}  
    where the convergence is locally uniform in $p$ on $(0,2)\cup(2,\infty)$.
\end{enumerate}
\end{theorem}
The pointwise convergence are derived using a similar technique as in \cite{Woener2011}, by making use of the convergence results for power variation derived in \cite{Todorov2010} and \cite{Corcuera2013}. The proofs are contained in Appendix \ref{appendix:proofs}.

Plotting the roughness signature function over different values of $p$ will give useful visual tool for determining the components of the process we are observing. When looking at low powers, the level of the roughness signature function will give the recirpocal of the activity level of the most active component. If this level is $1/2$, the most active component is a continuous martingale\footnote{This follows by the results of \cite{Todorov2010}}. If it is lower, the process contains a continuous rough part, where the level of the signature will be the Hurst roughness exponent of the process. Finally, if the level is higher than $1/2$, the process does not contain a continuous component, but is instead a pure-jump process with Blumenthal-Getoor index equal to the reciprocal of the level of the signature. The curve will be flat until reaching the activity level. If the process contains jumps, it will have a kink at the activity level and have slope $1/p$ afterward; otherwise, it will stay flat. This is illustrated in Figure \ref{fig:rsf_illustration}. 
\begin{figure}[H]
        \centering
        \includegraphics[scale=0.3]{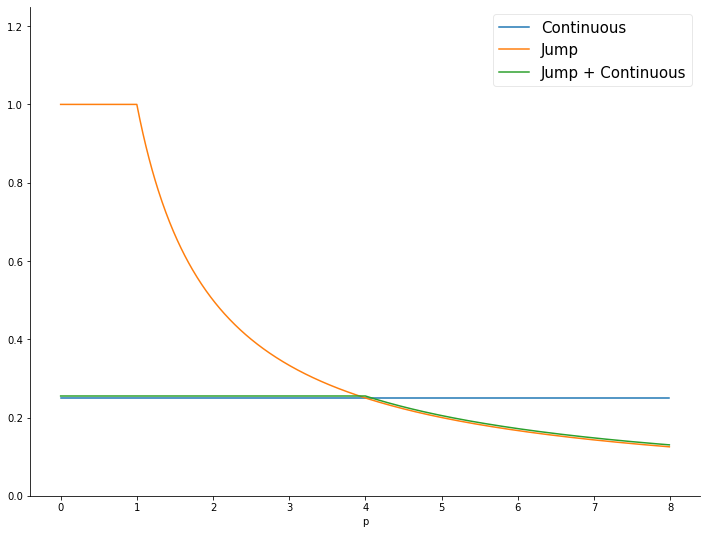}

    \caption{Illustration of examples of the roughness signature function. Continuous process with $\alpha=-0.25$, pure-jump process with activity index $1$, and continuous process plus jumps.}
\end{figure}\label{fig:rsf_illustration}
This is the tool that we will use for analysing our processes. When working with empirical data, we will typically have data for multiple time periods, e.g. data from multiple days or years, which can potentially contain jumps, and are hence susceptible to outliers. To deal with this, instead of looking at each separate roughness signature by itself, we will follow \cite{Todorov2010} and \cite{TodorovTauchen2011} and consider the quantile roughness signature function:
\begin{equation}
    H_q(\mathcal{X},p,\Delta_n)_T =\operatorname{Quantile}_q(\{H(\mathcal{X}_t,p,\Delta_n)\}_{t=1}^T)
\end{equation}
where we use the notation $H(\mathcal{X}_t,p,\Delta_n)$ to denote the roughness signature function for the observations of $\mathcal{X}$ over the period $t-1$ to $t$, while $\operatorname{Quantile}_q$ returns the $q$'th quantile of the input vector. In practice, we will be considering $q\in\{0.25,0.5,0.75\}$ to get a measure of the central tendency as well as the spread of the signature functions.

\section{Simulation Results}\label{sec:sim_study}

In this section, we will examine the finite-sample behaviour of the roughness signature function. To do this, we will simulate examples of $\mathcal{BSS}$-processes and of pure jump processes. As our example of a rough process, we will consider the gamma-$\mathcal{BSS}$ process. This is the process defined by \eqref{eq:cont_model}, where the kernel function is given by 
\begin{equation}\label{eq:gamma_BSS}
    g(x)=x^\alpha\exp(-\lambda x),\quad x>0,\;\alpha\in(-1/2,1/2),\;\lambda>0
\end{equation}
to simulate this process, we will use the hybrid scheme introduced by \cite{BLP2017}. We consider two different rough specifications, $\alpha=-0.25$, corresponding to moderate roughness, and $\alpha=-0.4$, corresponding to very rough sample paths. In addition, we will consider the case $\alpha=0$, which puts us in the semimartingale case of \cite{Todorov2010}---in fact, this will exactly result in a standard Ornstein-Uhlenbeck process. For simplicity, we will  take $\sigma_1\equiv1$, which means we will have no stochastic volatility, and the process will be Gaussian. 

For the jump processes, we will both consider a Poisson process, which is of finite activity and hence has Blumenthal-Getoor index 0, and a tempered stable process, where a parameter controls the index. 

The tempered stable process is a Lévy process with Lévy measure of the form 
\begin{equation}
    \nu(dx)=\begin{cases}
c_{1}e^{-\lambda_1x}x^{-1-\beta}dx & \text{on }(0,\infty)\\
c_{2}e^{-\lambda_2\vert x\vert}\vert x\vert^{-1-\beta}dx & \text{on }(-\infty,0)
\end{cases}
\end{equation}

Where $\lambda_1,\lambda_2>0$ are positive parameters, while $\beta\in(0,2)$ is a parameter determining the Blumenthal-Getoor index. 

For simulation of the tempered stable processes, we use the shot-noise representation derived in \cite{Rosinski2007}, and truncate the series at a large number. For each process, we simulate $R=500$ replications of the process over $1$ time period with $n\in\{500,2000\}$ observations per time period. We plot the median roughness signature over the replications as well as the 25- and 75th percentile in the sample, to get a feeling for the variation among the samples. 

First, we consider the purely continuous gamma-$\mathcal{BSS}$ process. Here, we consider $\alpha\in\{-0.25,-0.4\}$, corresponding to Hurst roughness index $H\in\{0.1,0.25\}$. Here, we get the following. 

\begin{figure}[H]
    \centering
    \includegraphics[scale=0.3]{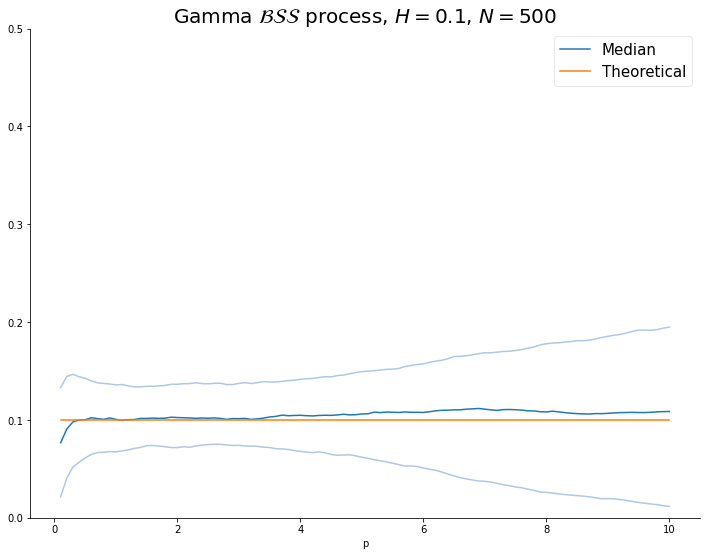}     \includegraphics[scale=0.3]{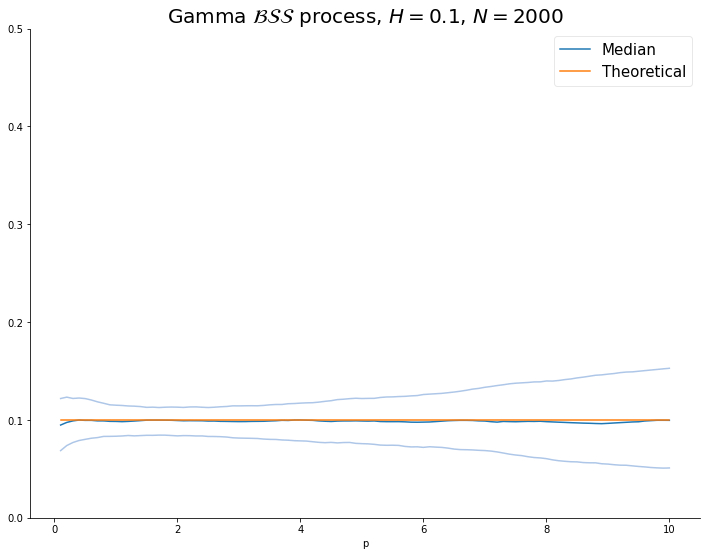}
    
    \includegraphics[scale=0.3]{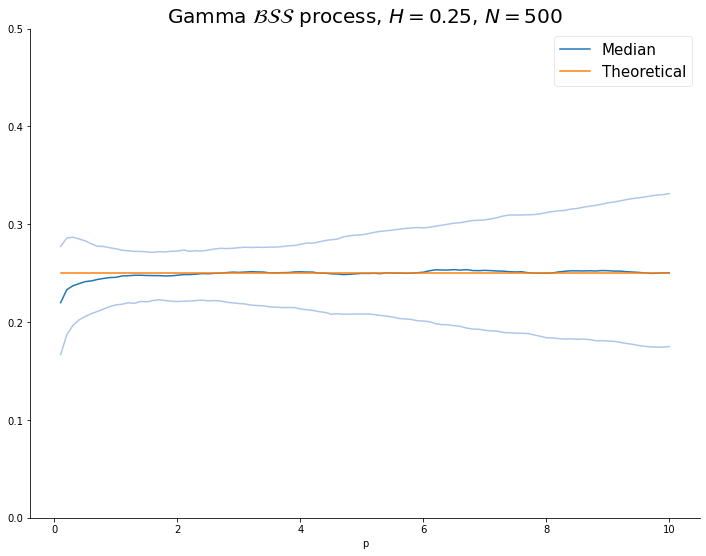}     \includegraphics[scale=0.3]{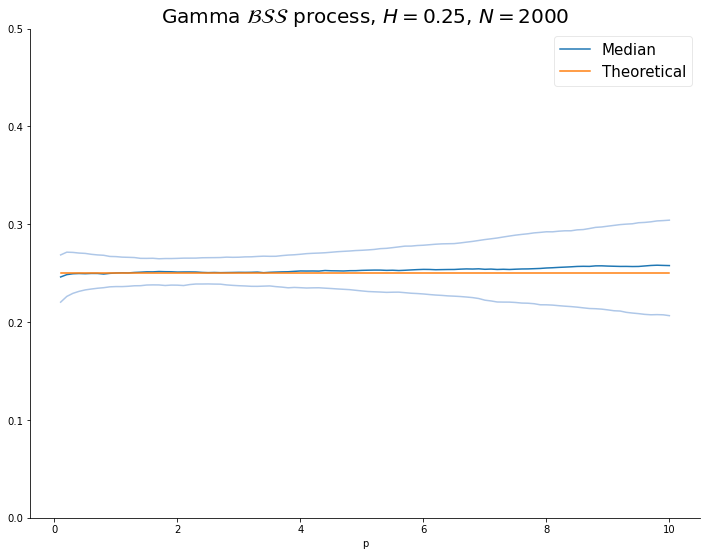}
\caption{Roughness signature functions for the gamma $\mathcal{BSS}$ process with roughness index $\alpha\in\{-0.25,-0.4\}$. The light-coloured blue lines are the 25- and 75th percentile, the hard blue line is the median signature function while the orange one is the theoretical limit.}
    \label{fig:quasf_ts_cof}
\end{figure}
We see that the roughness signatures behaves almost exactly as we would expect them to based on the asymptotic results. We also see that we do not need a huge sample for the effect to show.

Next, we consider the pure jump process, where we consider $\beta\in\{0.5,1,1.5\}$. 
\begin{figure}[H]
    \centering
    \includegraphics[scale=0.3]{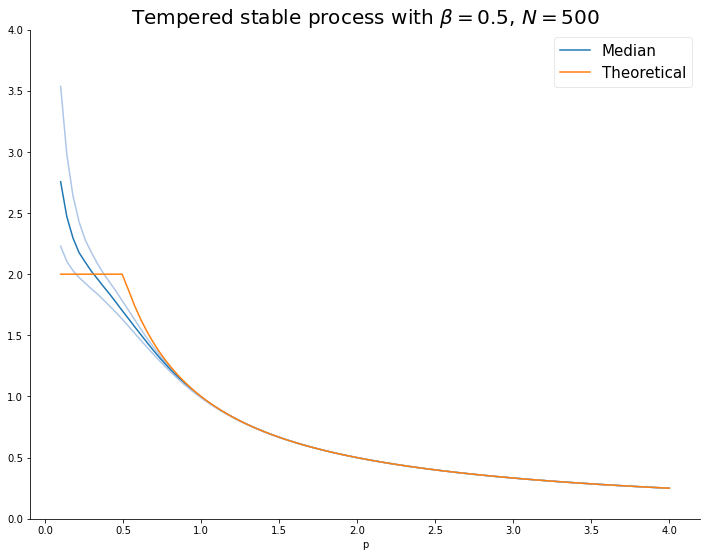}\includegraphics[scale=0.3]{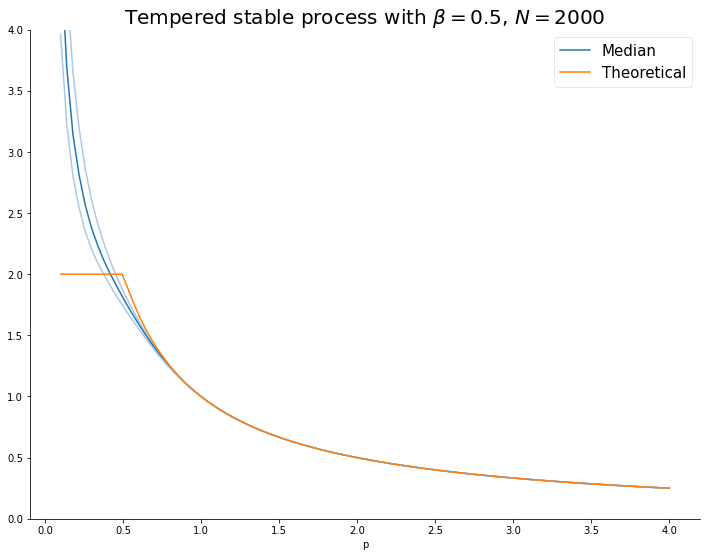}   
    
    \includegraphics[scale=0.3]{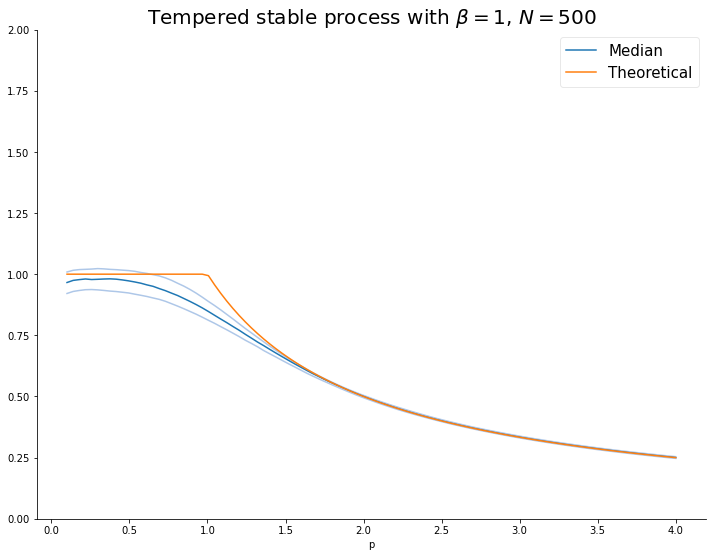} \includegraphics[scale=0.3]{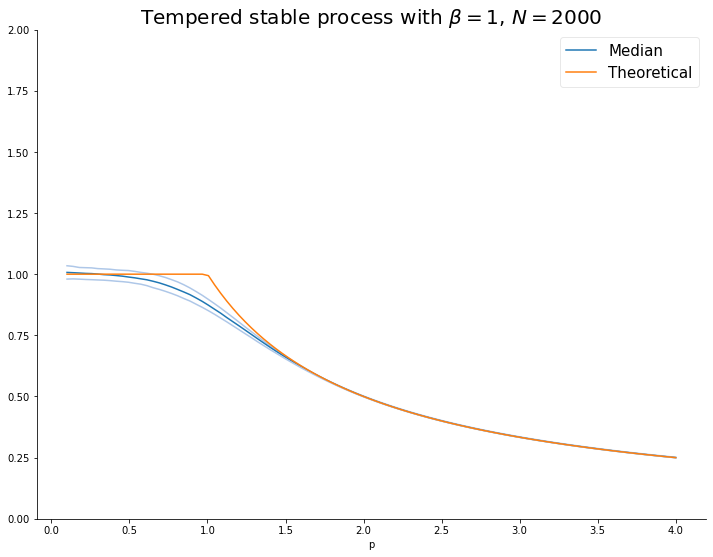}
    
    \includegraphics[scale=0.3]{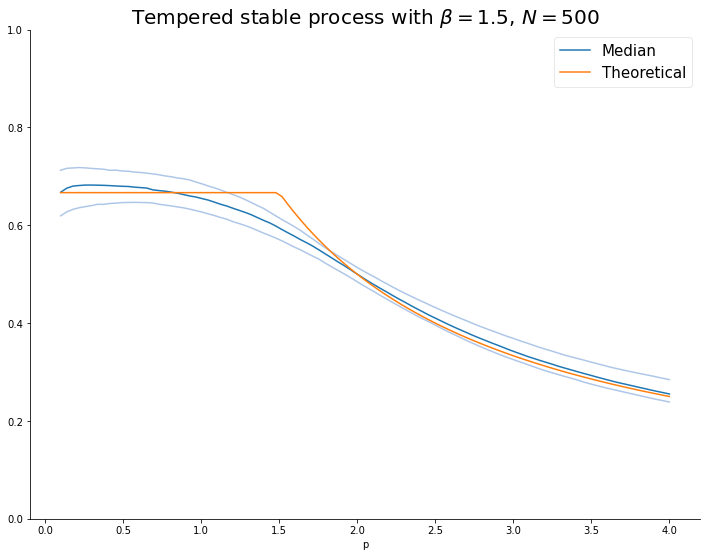}
    \includegraphics[scale=0.3]{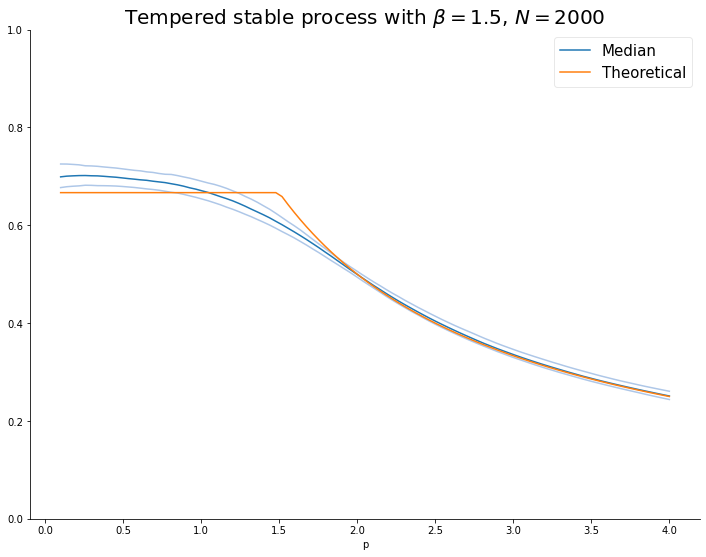}   
\caption{Roughness signature functions for the tempered stable process with Bluementhal-Getoor index $\beta\in\{0.5,1,1.5\}$. The light-coloured blue lines are the 25- and 75th percentile, the hard blue line is the median signature function while the orange one is the theoretical limit.}
    \label{fig:quasf_ts_cof}
\end{figure}
We see that the roughness signature function don't quite get the sharp cut-off at the Blumenthal-Getoor index. This is expected, as we are working with a finite sample, and the functionals we are considering are naturally smooth in finite samples. The performance does however to improve as we increase the number of observations in the time period. Hence, the simulations seems to confirm the asymptotic results.

Finally, we want to consider a combination of a continuous process and a pure-jump process. Here, we consider the gamma $\mathcal{BSS}$ processes with $\alpha=-0.25$ from above, but add a jump process. For the asymptotic results, the choice of jump process should not matter, so we simply add Poisson jumps to the process. The result is the following: 

\begin{figure}[H]
    \centering
    \includegraphics[scale=0.3]{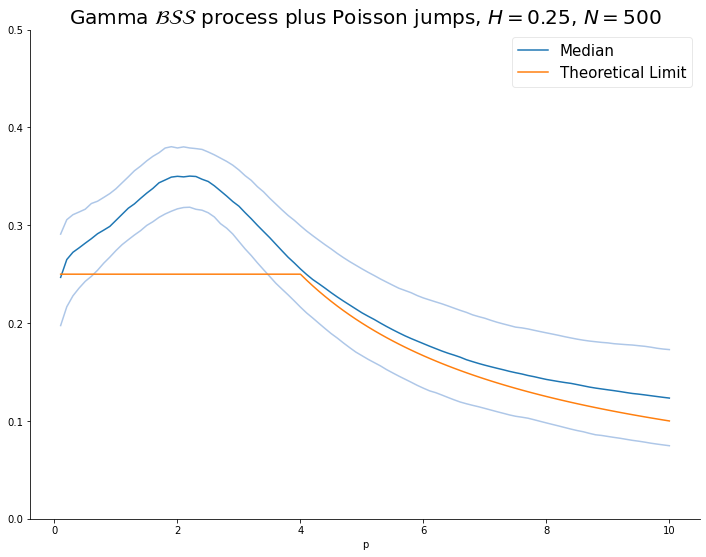}
    \includegraphics[scale=0.3]{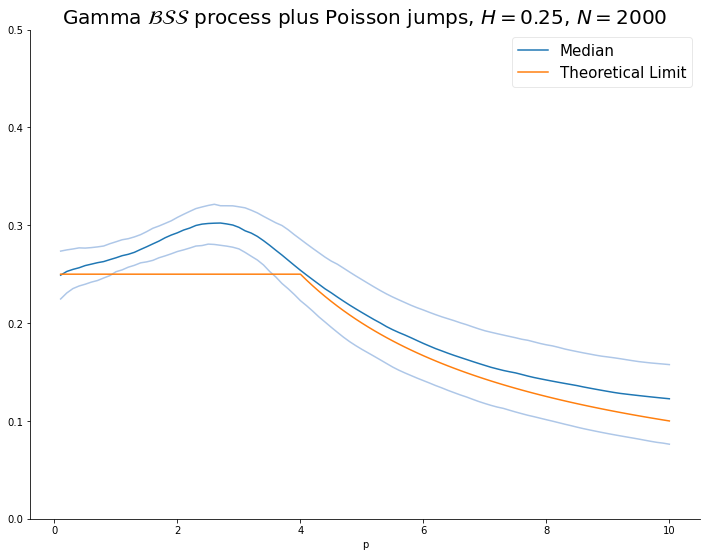}
\caption{Roughness signature functions for the gamma $\mathcal{BSS}$ process with roughness index $\alpha =0.25$ with rare Poisson jumps. The light-coloured blue lines are the 25- and 75th percentile, the hard blue line is the median signature function while the orange one is the theoretical limit.}
    \label{fig:quasf_ts_cof}
\end{figure}
Here, we fit is not perfect, but it does seem to improve as $n$ increases. We could also do the same for the rougher gamma $\mathcal{BSS}$ process, but here we would have to look at powers higher than $10$ to see the "kink" indicating jumps, and at these very high powers, the roughness signature function needs a much higher number of observations to approach its theoretical limit.  

Overall, the roughness signature function seems to work as we would expect from the asymptotic results. 

\section{Empirical Application}\label{sec:emp_application}

\subsection{Different Volatility Measures}
Through Theorem \ref{thm:consistency} and the simulation study, we have now established that the roughness signature function works as expected on rough processes and processes with jumps. We are thus ready to turn to an empirical application of the signature function. As mentioned in the introduction, our primary motivation for discerning rough processes from jumps processes and measuring the activity or roughness is to examine the processes driving volatility in the financial markets. This is interesting because it will allow us to determine if volatility is rough, as proposed by \cite{Gatheral2018}, and whether, in addition, there is a jump component present. A challenge when working with volatility is its inherent latent nature. There is no direct way to observe volatility, and we, therefore, need to choose which volatility measures to use to examine this question. We will examine if volatility is rough by plotting the estimated roughness signature function for three different volatility measures and examining whether the shape and level imply a rough process. As we have studied in the simulations in Section \ref{sec:sim_study}, we would expect a flat curve at a level lower than $1/2$ if the hypothesis of rough volatility is true.

\subsection{The Data}\label{section:data}

We examine three different measures of volatility: realized volatility (RV), the CBOE volatility index (VIX), and the recently introduced option-extracted spot volatility (OEV) of \cite{Todorov2019}.

The three series used in this paper are all measures of the volatility of the S\&P500 index. They all contain daily measures of the volatility for all trading days, which amounts to an average of 252 observations per year. 

The data on VIX and OEV is obtained from TailIndex.com. They provide daily observations from the start of 2008 to the end of 2020, resulting in a total of 13 years of data. The first three years of data, $2008$ through to $2010$, are calculated based on regular S\&P500 index options, while all data onward are calculated based on weekly S\&P 500 index options. 

The data on realized variance was obtained from the Oxford-Man Institute's realized library\footnote{https://realized.oxford-man.ox.ac.uk/}. The website computes a number of non-parametric volatility estimators on more than 30 stock indices, including the S\&P500. We use the realized variance, 5 min subsampled data, which gives daily observations on all trading days from the start of the year 2000 up to the present day. We capped our data at the end of December 2021 to only have full years in our sample. In Figure \ref{fig:time_series_of_data} we have provided plots of the three time series. 

\begin{figure}[H]
    \centering
    \includegraphics[scale=0.25]{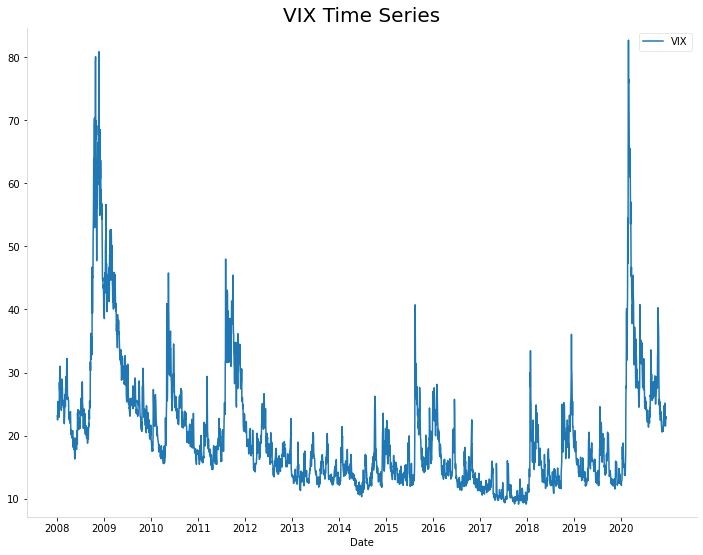}\includegraphics[scale=0.25]{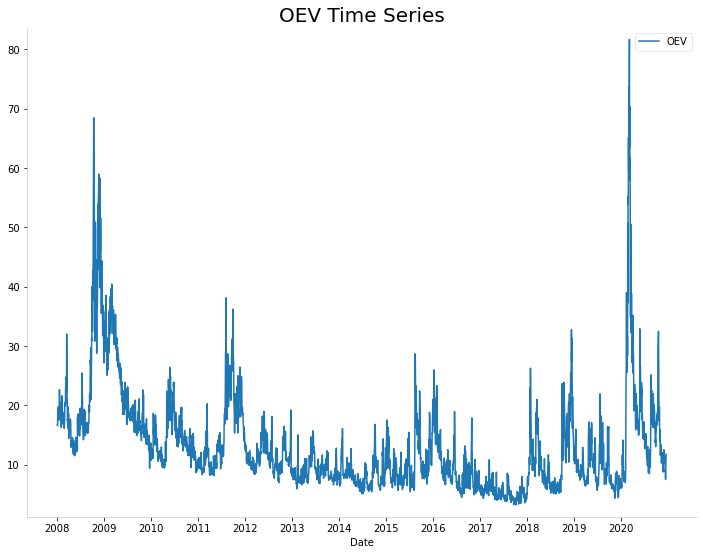}
    \includegraphics[scale=0.25]{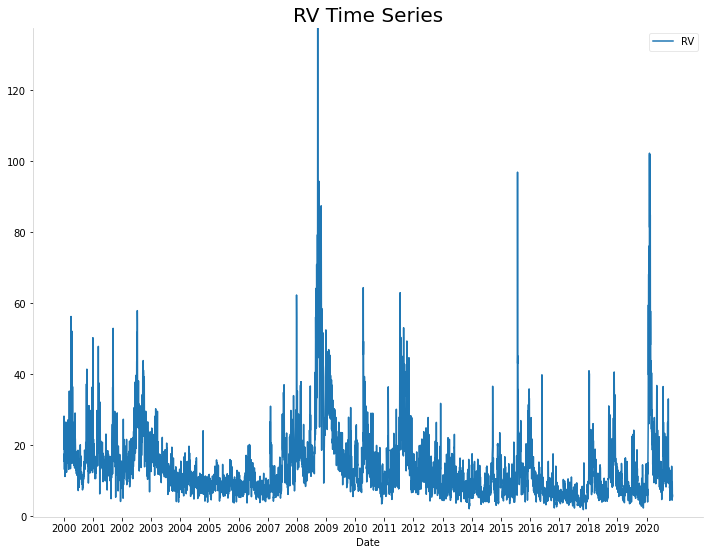}\includegraphics[scale=0.25]{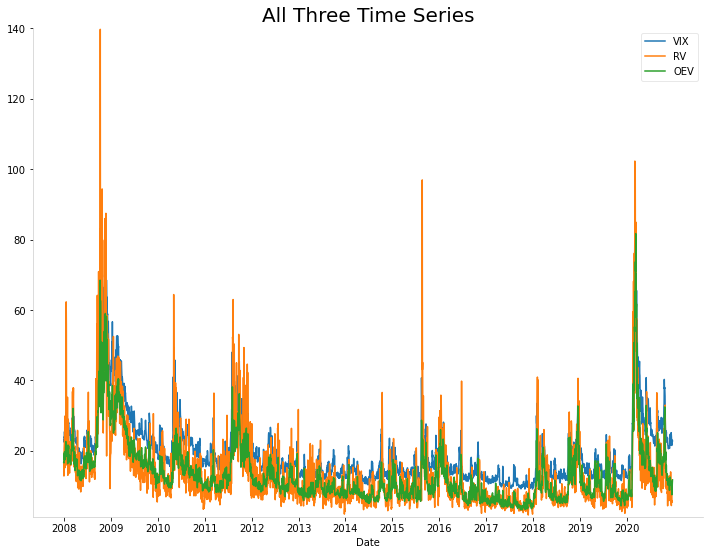}
\caption{Time series plot of the three volatility series. The top plots and bottom left plot are the full individual series, the last one is the three series in the time they overlap.}
    \label{fig:time_series_of_data}
\end{figure}

We see that the three measures are quite different. Firstly, we note that the VIX has a higher mean value than the other two assets. This might be because the VIX is a measure of what the market expects the future volatility to be rather than a measure of current/past volatility. It could be the case that investors on average expect more volatility than the average volatility level. Secondly, we note that realized variance seems to be a much rougher or noisier measure than the other two, which are evident in the bottom right plot of Figure \ref{fig:time_series_of_data}. This can also be seen in the descriptive statistics, as the RV has a higher standard deviation and higher kurtosis than the OEV series. OEV being a less noisy measure than RV is also consistent with the findings in \cite{Todorov2019}.

\subsection{Empirical Roughness Signature Functions}\label{section:empirical_asf}

We are now ready to examine the empirical roughness signature functions. For this, we will be computing the roughness signature functions using each year as one time period and will then be plotting the median as well as the $25$ and $75$ percent quantiles of it as explained in Section \ref{sec:ASF}.

In \cite{Todorov2010}, they restricted their attention to semimartingales, and as a result, the range $p\in(0,4)$ would always contain all the information on the process, as the maximum activity level of these processes is $2$. As we have seen above, when working with processes that are rougher than a Brownian motion, the activity level can be arbitrarily high, and we, therefore, need to examine higher powers. In the following, we have illustrated the roughness signature functions for $p\in(0,10)$.\footnote{We also examined the signature functions for higher powers for the realized volatility series, which looks the most rough, but found no further information in the higher powers and will thus not include these plots in this paper.}  

\begin{figure}[H]
    \centering
    \includegraphics[scale=0.3]{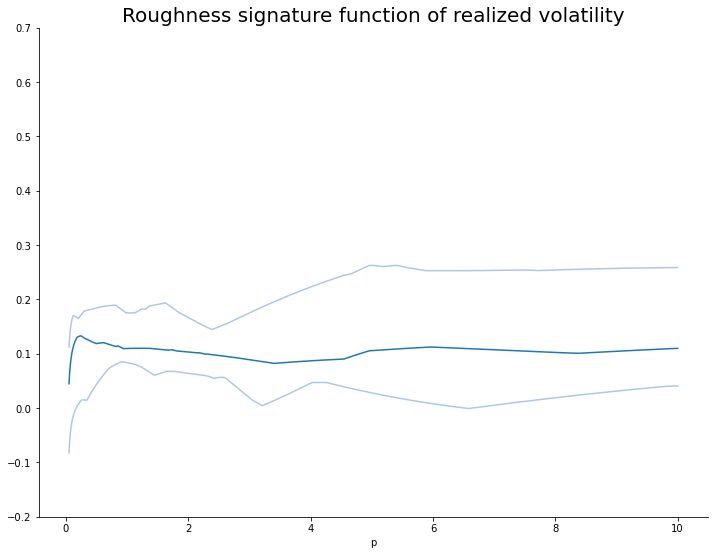}
    \includegraphics[scale=0.3]{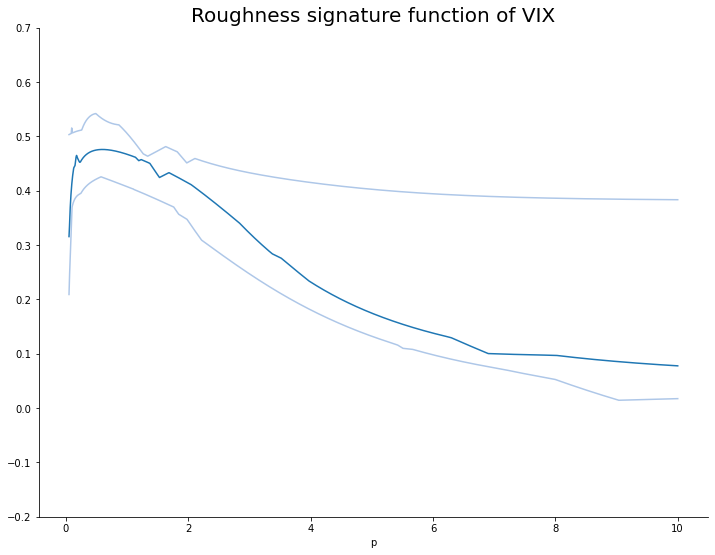}
    \includegraphics[scale=0.3]{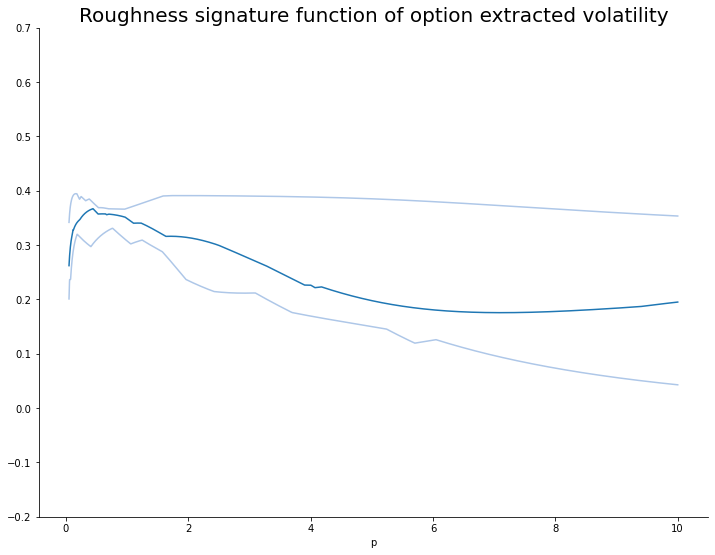}
\caption{Empirical quantile roughness signature functions for the three volatility measures. The curve in the dark shade is the median, while the curves in the light shade are the 25- and 75\% quantiles.}
    \label{fig:empirical}
\end{figure}

We see that there are quite big differences between the roughness signature functions of the three series. If we first focus on the realized volatility, a measure used in previous studies of the roughness of volatility, we see that we get exactly the results we would expect. The median curve is relatively flat at $0.1$, which is pretty close to the former values of roughness observed in the literature. If jumps were present, the curve would gain a slope past $p=10$, but as stated above, we did not find this during the analysis. If we consider the OEV, we see that the series also shows signs of roughness but much less than the realized volatility. This is most possibly because RV is a noisier measure than the OEV. A Hurst exponent of $0.5$, corresponding to an activity level of $2$ (i.e. a continuous semimartingale) is however still outside quartiles. We note that the roughness signature function is not quite as flat as for the RV series, but it is still far from being as sloped as it would be in the presence of jumps. We, therefore, believe the slightly curved shape is simply due to sheer randomness, seeing as our sample only contains $13$ years of daily data. 

If we then turn to the roughness signature function of the VIX, we see that it does not imply roughness, unlike the roughness signature functions of the two other measures. In fact, the plot seems very consistent with the inverse roughness signature function of a continuous martingale plus jumps. To see this, we have added the theoretical roughness signature function alongside the VIX plot: 

\begin{figure}[H]
    \centering
    \includegraphics[scale=0.3]{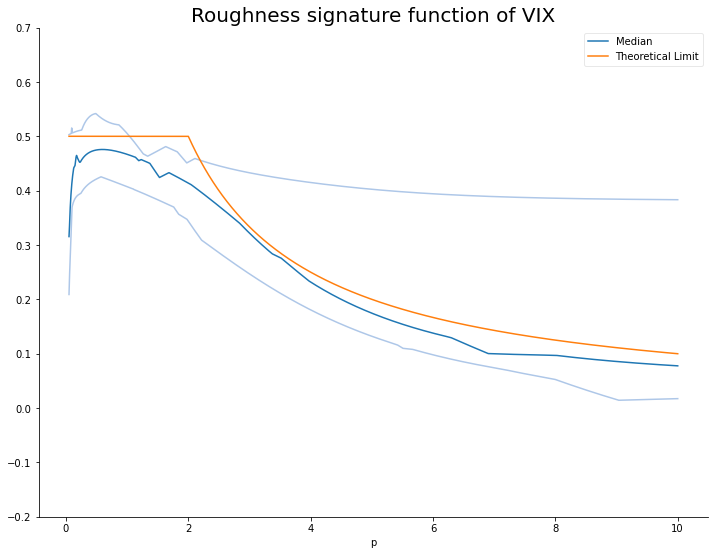}
\caption{Empirical roughness signature functions for the VIX, with the theoretical roughness signature function of a continuous martingale plus jumps imposed. The curve in the dark shade is the median, while the curves in the light shade are the 25- and 75\% quantiles. The orange curve is the theoretical limit of a continuous semimartingale with jumps.}
    \label{fig:vix_cm}
\end{figure}

We see that this fits very well, with the theoretical curve almost entirely within the quartiles for the roughness signature function. This suggests that while realized volatility is rough, as we have seen above, the VIX could actually be smooth, and while the spot volatility contained no evidence of jumps, the VIX does. 

We have seen evidence of the realized volatility and the option-extracted volatility being rough and will now compute point estimates of the level of roughness. We will do this using a point estimator inspired by the activity point estimator of \cite{Todorov2010}. The idea is to use that the roughness signature function stays flat up $p=H$, where $H=\frac{1}{1/2+\alpha}$ is the Hurst roughness exponent. Furthermore, we know that this convergence is uniform. This suggests that for a sufficiently small $\tau$, the following estimator should be consistent. 
\begin{equation} \label{eq:point_est_H}
    \hat{H} = \frac{1}{\hat{H}^{-1}(X,\tau,\Delta_n)-\tau}\int_{p=\tau}^{p=\hat{H}^{-1}(X,\tau,\Delta_n)}\hat{H}(X,p,\Delta_n)dp
\end{equation}

For this to work, $\tau$ needs to be a small number, less than the inverse of the Hurst exponent of the process. Given the low level of the plots and the slightly irregular behaviour near zero, we use $\tau=0.5$. For the realized volatility, we get $\hat{H}=0.09$, while for the OEV we get $\hat{H}=0.39$. We see that these point estimates confirm what the initial visual inspection of the roughness plots suggested.


Up to this point in the analysis, we have analyzed the roughness of the volatility series directly. In most applications, it is not assumed that volatility itself is modelled as a rough process such as a $\mathcal{BSS}$-process, but rather that the logarithm of volatility is. This is what was done in the seminal paper by \cite{Gatheral2018}, as well as in most followup papers including \cite{Bayer2016a}, \cite{Bennedsen2021}, \cite{BCSV2020}, and \cite{BennedsenEtAl2023}. This is done to ensure that the volatility process stays positive. From Theorem 2.1 of \cite{Bennedsen2021} we do, however, know that under some assumptions, the volatility process will inherit the roughness of the log-volatility process, and the results should therefore be the same. We wish to confirm that this is the case, so we have carried out an identical analysis to the one above for the log of each volatility series. The resulting roughness signature functions look very similar to the ones for the original series, and the point estimates also agree nicely with the ones from the original volatility series, confirming that the volatility series inherits the roughness of the log-volatility series. Due to the similarity, we have decided not to report these results. 

\subsection{Discussion}\label{section:discussion}

From the analysis above, we have two main takeaways. Volatility is rough, but when it is extracted from the options market using the spot volatility estimator of \cite{Todorov2019}, it looks smoother than when extracted from returns using realized volatility. Further, despite volatility being rough, the VIX index behaves as a smooth process with jumps. This section will consider potential explanations of these two somewhat surprising findings and considers which of the three series are best suited for measuring the roughness of volatility. 

One reason why the realized volatility series displays rougher paths than the option extracted volatility might be due to noise in the observations. It is well known when we sample high-frequency returns, what we observe is generally not increments of the efficient price process, but it will instead be contaminated by microstructure noise. This would typically be modelled as an additive term in the log-price. We know that realized variance is not a noise-robust volatility measure. If microstructure noise is present, we expect that this will affect our estimates, as it has been documented by \cite{Bennedsen2020} that noise will bias smoothness downwards, i.e. it makes data look rougher than it is. As mentioned in an earlier section, the estimator used in that article is an alternative version of the \cite{Corcuera2013} estimator, and we would therefore expect it to affect our estimator similarly. The issue of measurement noise was also considered in the web-available appendix of \cite{Todorov2010} in the semimartingale setting. Here, they show that when additive noise is present, the activity signature function will diverge to infinity no matter the underlying model. The intuition is that the measurement errors will become the most active part of the process, as their magnitude is independent of the sampling frequency. In our setting, where the smoothness is the inverse of the activity, this will also corresponds to the data looking rougher than it actually is. One way to solve this would be to use a noise robust measure of volatility, e.g. the realized kernel estimators of \cite{Barndorff-Nielsen2008}.

However, as we have seen above, the realized variance does not converge to the spot volatility but rather to the integrated volatility over the given day. We know that integration is generally a smoothing operation, and hence we would expect the realized variance series - if it was measured without any noise - to look smoother than the true spot volatility series. This problem was also considered by \cite{Gatheral2018} who showed through simulation and analytically in a very simplified model that proxying spot volatility by its integrated value will lead to a substantial bias. They find that the longer the integration is done over, the bigger the bias. In the literature, we have seen two ways the smoothing problem can be solved. First, \cite{Bennedsen2020} considers ways to handle the noise present in high-frequency data and presents a noise-robust estimator of the roughness. This approach is highly relevant to us. If the roughness signature functions could be made noise-robust, it would be possible to use intraday volatility measures rather than the daily ones. Thus, we would no longer be integrating spot volatility over an entire day, and would expect the smoothing effect to diminish. Secondly, \cite{Fukasawa2019} and \cite{BCSV2020} develop methods for estimating the Hurst exponent where the dynamics of the integrated variance is used directly in place of the spot volatility dynamics, thus avoiding the bias from using the wrong dynamics when working with RV. Overall, for the realized variance, we thus have two sources of error pulling the roughness estimate in different directions. 

There are also (at least) two sources of error when using the option extracted volatility for measuring roughness. The first is that in the construction of the estimator, one of the key assumptions when deriving the convergence is that the spot volatility is a semimartingale process. There is an apparent inconsistency in using an estimator of volatility that builds on an assumption of smooth volatility for measuring the roughness of volatility. What we are really investigating is then whether volatility extracted using a Brownian model will still exhibit roughness. Ideally, it should be investigated whether this estimator is also valid in a rough volatility setting, which would allow us to interpret this as a proxy for the true spot volatility. The second source of error lies in the remaining time to maturity in the options used to span the characteristic function of the returns, and thus to extract the volatility. We know that under some assumptions, the estimator will converge to the spot volatility if the time to maturity and mesh of the strike grid goes to zero. In practice, we are not using options with zero time to maturity but instead estimating the volatilities using on prices of weekly options, typically with 2-5 days till maturity. This remaining time to maturity could be the reason for the smoothness. \cite{LivieriEtAl2018} have previously studied the roughness of option extracted volatility using the simpler Black-Scholes implied volatility and the Medvedev-Scaillet correction formula applied to implied volatilities of short-term at-the-money implied volatility \citep{Medvedev2011} and found estimates of the Hurst exponent right around $0.3$, very similar to our findings. They argue that this higher estimate of the Hurst exponent is due to a smoothing effect from the remaining time to maturity in the options they use for extracting the volatility. They illustrate the smoothing effect using a simulation study and derive it analytically in a very simplified setup. 

Inspired by their findings, we will also carry out a simulation experiment to study how the roughness of the volatility estimator of \cite{Todorov2019} behaves when the underlying volatility is rough, but the options used have non-zero time to maturity. We consider the same simple model, where volatility is modelled as an exponential fractional Brownian motion, while returns are driven by a Brownian motion. Our model is then
\begin{equation}
\begin{split}
    d\log S_t &= \sigma_tdZ_t \\
    d\log\sigma_t&=\eta dW^H_t
\end{split}
\end{equation}

Where $Z_t$ is a standard Brownian motion and $W^H$ is a fractional Brownian motion with Hurst exponent $H\in(0,1/2)$, independent of $Z_t$. We generate $M=1000$ paths of this for $T=1000$ time periods with $H=0.04$ to follow \cite{LivieriEtAl2018} and to be consistent the findings in literature that accounts for the upward-bias when using integrated variance rather than spot variance. We condition on the volatility up to time $0$, and then use these to price call options for a range of strikes and time-to-maturities. This leaves us with a time series of option prices for a range of strikes for each maturity. We can then compute a time series of the OEV volatility estimates for each maturity. We do this for $\tau\in\{1,2,3,4,5,7,10\}$ days to maturity of the options and then estimate the roughness using both the estimator (\ref{eq:point_est_H}) from earlier, where we start the integration at $0.1$, and the roughness signature function evaluated in $p=2$. We repeat this $10$ times and plot the mean estimates of $H$. 

The resulting estimates of $H$ are the following: 
\begin{figure}[H]
    \centering
    \includegraphics[scale=0.25]{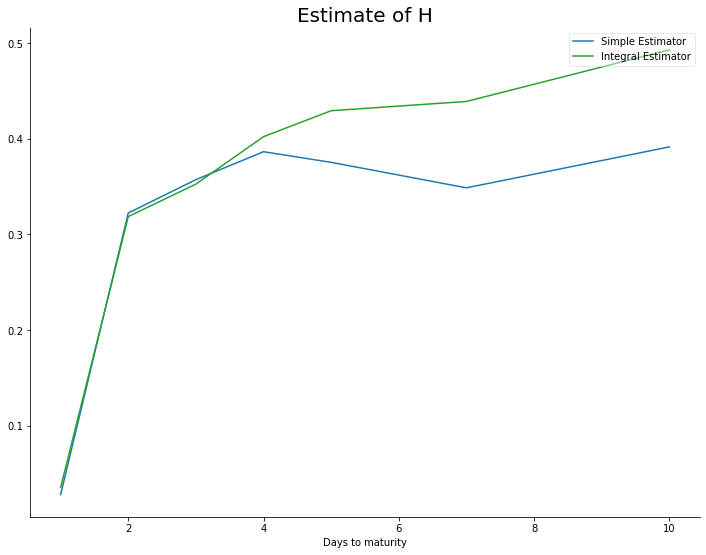}
\caption{Estimates of the Hurst exponent from the option implied volatility of \cite{Todorov2019} for different option maturities of simulated option prices. The simple estimator is the roughness signature function evaluated in $p=2$, while the integral estimator is (\ref{eq:point_est_H}) using $\tau=0.1$.}
    \label{fig:H_est_oev_sim}
\end{figure}

We see that in this simple model, we get very high estimates of $H$ despite the true volatility series being very rough, which supports the hypothesis that the OEV looks smoother than the underlying volatility. The days to maturity of the options used for our empirical OEV is between $2$ and $5$ days, and we see that our estimates of $H$ between $0.3$ and $0.4$ are quite consistent with the results of this simulation. We note that the simulation is done in an extremely simplified setting, where there are no jumps in returns, no leverage effect, and where volatility is a non-stationary process, but we still see it as a strong indicator of the reason for the smoothness. Between realized volatility and option extracted volatility, we thus have more faith in the estimated roughness from the realized volatility series. 

If we turn to the VIX, there are several reasons why it could behave like a smooth process with jumps. First of all, while the other two time series are simply volatility measures, the VIX is a traded asset. While it is not possible to take a direct position in the VIX itself, futures and options are liquidly traded on exchanges and are some of the most popular instruments for taking positions in pure volatility. We know from the first fundamental theorem of asset pricing that the absence of arbitrage is essentially equivalent to the existence of an equivalent martingale measure \citep{bjoerk2009}, which means that if the VIX market is free of arbitrage, we would expect the dynamics of the VIX to be semimartingale. This is exactly what we find. The VIX being smooth thus seems a natural consequence of the existence of a market with the VIX as the underlying. 

Another potential explanation lies in the construction of the VIX. We know that the VIX is constructed such that the squared VIX replicates the time-$t$ conditional expectation of the integrated variance over the coming 30 days. By construction, a conditional expectation will always be a martingale process, no matter what we are taking the expectation of. Therefore, we would expect the squared VIX to behave like a martingale, and since the square root is a concave function, Jensen's inequality immediately yields that the VIX should be a supermartingale - and therefore, in particular, a semimartingale. The observed smoothness of the VIX might then be independent of the roughness of the underlying volatility and whether or not there is a market, but it may result from the expression we are trying to span. 

Related to the above, it is also important to remember that the VIX is not meant to directly measure spot volatility but rather measure the expected integrated variance. We know that integration is a smoothing operator, so we suspect that both the integration inside the expectation and taking the expectation will make the series look smoother than it is. The integration inside will yield a smoothing effect similar to the one discussed for realized variance, while we expect the expectation to have a similar effect to the one we saw when looking at the OEV.

\section{Conclusion}\label{sec:conclusion}
In this paper, we have introduced the roughness signature function, which is a tool that can be used to asses the components of an observed stochastic process and the regularity of its sample path. We have derived consistency of the signature function, and through an extensive Monte Carlo study, we have seen that, given high-frequency observations of some stochastic process, the function is able to distinguish between observations generated by a rough process, i.e. a continuous process with roughness index of the sample paths of degree strictly less than $0$, a pure jump process, and a combination of the two. This has been examined using tempered stable processes and Brownian semi-stationary processes. We are able to recover the smoothness or activity of the given process, and when a process is made up of a combination of the two processes we can recover the activity of the ``most active'' part, which will always be the rough process in our setup.

We have applied these roughness signature functions to three different measures of the volatility of the S\&P500 index: realized volatility, the VIX, and the option extracted volatility of \cite{Todorov2019} We have seen that the three series seem to be driven by very different processes. The realized volatility and the option extracted volatility appear to be driven by rough processes, with the realized variance being a lot rougher than the option-extracted volatility. Despite this, the VIX appears to be driven by a continuous martingale with jumps. 

The roughness estimate of the realized volatility is mostly consistent with the existing literature, although a bit smoother than some recent studies, while the option-extracted measure suggests volatility to be a lot smoother than previous studies have found. We hypothesise that the option extracted volatility appears smooth due to a smoothing effect of the remaining time to maturity of the options used. We have demonstrated this smoothing effect of the estimator through a small simulation study, showing how a very rough volatility process can look almost Brownian when extracted with the \cite{Todorov2019} estimator using options with non-zero time to maturity. We believe that the role of the VIX as the underlying asset of many traded products might make it behave like a semimartingale to exclude arbitrage. We also hypothesise that its semimartingale behaviour might be because the squared VIX is constructed to span a martingale, implying the VIX should be a supermartingale.

Going forward, we believe it would be interesting to study why we see such large differences between roughness signature functions of the three measures. Is the reason why the VIX is smooth only because the VIX acts as the underlying for a traded asset, because the squared VIX spans a martingale, or is there some other explanation? 

Given the large differences between the three different volatility measures of the same underlying, it would also be interesting to explore the roughness signature functions of other volatility measures. As we have seen, the realized volatility measure is neither robust to noise nor jumps in the price, and it would therefore be interesting to use a measure designed to handle these. While jumps are not likely to contribute a lot in our setting, as jumps in the S\&P500 are typically not very large, this could be very important if looking at e.g. individual stocks. It would also be interesting to examine the same measures as in this project but applied to different indices or different asset classes to see if the shape of the roughness signature plots is dependent on the underlying. 

Finally, in this paper, we have used daily volatility measures, some of which are based on high-frequency returns. Ideally, we would want to be able to use high-frequency measures of volatility instead - i.e. intraday volatility measures. We know that high-frequency financial data is typically contaminated by measurement noise which in the current setup will bias the estimated activity index upwards and thus bias the estimated roughness downwards. To use high-frequency observations, we would thus need to make the roughness signature function noise-robust. To do this, it would be interesting to explore the approach developed in \cite{Bennedsen2020}, where they exactly show how to tackle additive noise when estimating roughness. 

\pagebreak

{\small 
\bibliographystyle{chicago}
\bibliography{mybib-v3.bib}
}

\pagebreak
\appendix
\begin{section}{Proof of asymptotic results}\label{appendix:proofs}
Let $(\Omega,\mathcal{F},(\mathcal{F}_t)_{t\geq0},\mathbb{P})$ be a filtered probability space satisfying the usual conditions. We assume that we are observing a stochastic process $\mathcal{X}=(\mathcal{X}_t)_{t\in[0,T]}$ at the discrete time points $\{0,\Delta_n,2\Delta_n,...,n\Delta_n = T\}$. We assume that there are three potential underlying processes for generating $\mathcal{X}$, namely

The continuous model: 
\begin{equation} \label{eq:cont_model}
    X_{t}=\int_{-\infty}^{t}g(t-s)\sigma_{1s}dW_s
\end{equation}

The pure-jump model: 
\begin{equation} \label{eq:jump_model}
    Y_{t}=\int_{0}^{t}b_{2s}ds+\int_{0}^{t}\int_{\mathbb{R}}\sigma_{2s-}\kappa(x)\tilde{\mu}(ds,dx)+\int_{0}^{t}\int_{\mathbb{R}}\sigma_{2s-}\kappa'(x)\mu(ds,dx)
\end{equation}

and the continuous plus jumps model: 
\begin{equation} \label{eq:cont_jump_model}
    Z_{t}=X_{t}+Y_{t}
\end{equation}
Here, $W=(W_t)_{t\in\mathbb{R}}$ is a two-sided Brownian motion, $g:\mathbb{R}^+\to\mathbb{R}$ is a deterministic weight-function such that $g\in L^2(\mathbb{R}^+)$, and $\sigma_1=(\sigma_{1t})_{t\in\mathbb{R}}$ and $\sigma_2=(\sigma_{2t})_{t\in\mathbb{R}^+}$ are cádlág processes. $\mu$ is a jump measure on $\mathbb{R}_{+}\times\mathbb{R}$ with compensator $\nu(x)dxds$, $\tilde{\mu}(ds,dx):=\mu(ds,dx)-\nu(x)dxds, \kappa'(x)=x-\kappa(x)$ where $\kappa(x)$ is a continuous truncation function.

We then want to prove the convergence of the roughness signature function, which is given by 

\begin{equation}\label{eq:cof_expr}
    \begin{split}
        \hat{H}(\mathcal{X},p,\Delta_n)_T &= \frac{\log(V(\mathcal{X},p,1,2;\Delta_n)_T) - \log(V(\mathcal{X},p,1,1;\Delta_n)_T)}{\log(2)p} \\
        &=  \frac{\log(\sum_{i=2}^n\vert \mathcal{X}_{i\Delta_n}-\mathcal{X}_{(i-2)\Delta_n}\vert^p)-\log(\sum_{i=1}^n\vert \mathcal{X}_{i\Delta_n} - \mathcal{X}_{(i-1)\Delta_n}\vert^p)}{\log(2)p}
    \end{split}
\end{equation}
where $V(\mathcal{X},p,k,\nu;\Delta_n)_T$ is the $p$'th power variation of $X$ of order $k$, sampled at frequency $\nu\Delta_n$. 

To prove the convergence, we first need to establish the following lemma:

\begin{lemma}
Let $(\mathcal{X}_t)_{t\in[0,T]}$ be a stochastic process, and assume that for some $k\in\mathbb{R}$ and $p\in(a,b)$ 
\begin{equation}
    \Delta_n^{1-pk}V(\mathcal{X},p,1,1;\Delta_n)_T\overset{u.c.p.}{\longrightarrow}\Xi,
\end{equation}
for $\Delta_n\to0$ for a random variable $\Xi$ such that $0<\Xi<\infty$ a.s.. Then 
\begin{equation}
    (2\Delta_n)^{1-pk}V(\mathcal{X},p,1,2;\Delta_n)_T\overset{u.c.p.}{\longrightarrow}2\Xi,
\end{equation}
\end{lemma}
\begin{proof}
Assume for simplicity that $n$ is even. To see this, we first rewrite $V(\mathcal{X},p,1,2;\Delta_n)$ as 
\begin{equation}
    \begin{split}
        V(\mathcal{X},p,1,2;\Delta_n)_T &= \sum_{i=2}^n\vert \mathcal{X}_{i\Delta_n}-\mathcal{X}_{(i-2)\Delta_n}\vert^p \\
        &= \sum_{i=1}^{n/2}\vert \mathcal{X}_{2i\Delta_n}-X_{2(i-2)\Delta_n}\vert^p + \sum_{i=1}^{n/2-1}\vert \mathcal{X}_{(2i+1)\Delta_n}-X_{(2(i-1)+1)\Delta_n}\vert^p
    \end{split}
\end{equation}
Here, each of the two sums are increments of $\mathcal{X}$ over the interval $[0,T]$ sampled at frequency $2\Delta_n$. The major difference between these sums and the first one is that we now have non-overlapping intervals. This means that each of these sums will correspond to $V(\mathcal{X},p,1,1;2\Delta_n)_T$, where the only difference is whether the sums is started at time $t=0$ or at time $t=\Delta_n$. As $n\to\infty$, this will naturally make no difference, as we are working with either a continuous process or one with Lévy-driven jumps, and hence the $p$-variation over an increasingly small interval must tend to zero. By our assumption we thus get that 

\begin{equation}
    \begin{split}
        &(2\Delta_n)^{1-pk}V(\mathcal{X},p,1,2;\Delta_n)_T \\
        &= (2\Delta_n)^{1-pk}(\sum_{i=1}^{n/2}\vert \mathcal{X}_{2i\Delta_n}-\mathcal{X}_{2(i-2)\Delta_n}\vert^p + \sum_{i=1}^{n/2-1}\vert \mathcal{X}_{(2i+1)\Delta_n}-\mathcal{X}_{(2(i-1)+1)\Delta_n}\vert^p) \\
        &= (2\Delta_n)^{1-pk}\sum_{i=1}^{n/2}\vert \mathcal{X}_{2i\Delta_n}-\mathcal{X}_{2(i-2)\Delta_n}\vert^p + (2\Delta_n)^{1-pk}\sum_{i=1}^{n/2-1}\vert \mathcal{X}_{(2i+1)\Delta_n}-\mathcal{X}_{(2(i-1)+1)\Delta_n}\vert^p \\
        &\overset{u.c.p.}{\longrightarrow} \Xi + \Xi \\
        &= 2\Xi
    \end{split}
\end{equation}
\end{proof}

Having established this, we are ready to show the main proposition that will be used for showing the convergence. 

\begin{proposition}\label{prop:convergence_general}
Let $(X_t)_{t\in[0,T]}$ be a stochastic process and assume that for some $k\in\mathbb{R}$ and some $p\in(a,b)$ that
\begin{equation}\label{eq:convergence_condition}
    \Delta_n^{1-pk}L_0(\Delta_n)V(\mathcal{X},p,1,1;\Delta_n)_T\overset{u.c.p.}{\longrightarrow}\Xi,
\end{equation}
for a random variable $\Xi$ such that $0<\Xi<\infty$ and function $L_0$ which is slowly varying at zero, then 
\begin{equation}
    \hat{H}(\mathcal{X},p,\Delta_n)_T \overset{u.c.p.}{\longrightarrow} \begin{cases}
k & \text{if }pk\neq1\\
1/p & \text{if }pk=1
\end{cases}
\end{equation}

\end{proposition}

\begin{proof}

To see this, we first expand the estimator

\begin{equation}
    \begin{split}
        &\hat{H}(\mathcal{X},p,\Delta_n)_T \\
        &= \frac{\log(V(\mathcal{X},p,1,2;\Delta_n)_T) - \log(V(\mathcal{X},p,1,1;\Delta_n)_T)}{\log(2)p} \\
        &= \frac{\log(2\Delta_nL_0(2\Delta_n)V(\mathcal{X},p,1,2;\Delta_n)_T) - \log(\Delta_nL_0(\Delta_n)V(\mathcal{X},p,1,1;\Delta_n)_T)-\log(\frac{2\Delta_n}{\Delta_n}))}{\log(2)p}\\
        &+ \log\left(\frac{L_0(\Delta_n)}{L_0(2\Delta_n)}\right)\frac{1}{\log(2)p} \\
        &= \frac{\log(2\Delta_nL_0(2\Delta_n)V(\mathcal{X},p,1,2;\Delta_n)_T) - \log(\Delta_nL_0(\Delta_n)V(\mathcal{X},p,1,1;\Delta_n)_T)-\log(2)}{\log(2)p} \\
        &+ \log\left(\frac{L_0(\Delta_n)}{L_0(2\Delta_n)}\right)\frac{1}{\log(2)p} \\
        &= \frac{\log(2\Delta_nL_0(2\Delta_n)V(\mathcal{X},p,1,2;\Delta_n)_T) - \log(\Delta_nL_0(\Delta_n)V(\mathcal{X},p,1,1;\Delta_n)_T)}{\log(2)p} \\
        &- \frac{1}{p} + \log\left(\frac{L_0(\Delta_n)}{L_0(2\Delta_n)}\right)\frac{1}{\log(2)p} \\
        &= \frac{\log\left(\dfrac{\left(2\Delta_n\right)^{pk}}{\left(\Delta_n\right)^{pk}}\right) + \log(\left(2\Delta_n\right)^{1-pk}L_0(2\Delta_n)V(\mathcal{X},p,1,2;\Delta_n)_T) - \log(\left(\Delta_n\right)^{1-pk}L_0(\Delta_n)V(\mathcal{X},p,1,1;\Delta_n)_T)}{\log(2)p} \\
        &- \frac{1}{p} + \log\left(\frac{L_0(\Delta_n)}{L_0(2\Delta_n)}\right)\frac{1}{\log(2)p} \\
        &= \frac{pk\log(2) + \log(\left(2\Delta_n\right)^{1-pk}L_0(2\Delta_n)V(\mathcal{X},p,1,2;\Delta_n)_T) - \log(\left(\Delta_n\right)^{1-pk}L_0(\Delta_n)V(\mathcal{X},p,1,1;\Delta_n)_T)}{\log(2)p} \\
        &- \frac{1}{p}+ \log\left(\frac{L_0(\Delta_n)}{L_0(2\Delta_n)}\right)\frac{1}{\log(2)p} \\
        &= \frac{pk}{p}+ \frac{\log(\left(2\Delta_n\right)^{1-pk}L_0(2\Delta_n)V(\mathcal{X},p,1,2;\Delta_n)_T) - \log(\left(\Delta_n\right)^{1-pk}L_0(\Delta_n)V(\mathcal{X},p,1,1;\Delta_n)_T)}{\log(2)p} \\
        &- \frac{1}{p} + \log\left(\frac{L_0(\Delta_n)}{L_0(2\Delta_n)}\right)\frac{1}{\log(2)p}
    \end{split}
\end{equation}

Here, we have just used basic properties of the logarithm. Now, by assumption we know that $\left(\frac{1}{n}\right)^{1-pk}V(\mathcal{X},p,1,1;\Delta_n)_T\overset{u.c.p.}{\to}\Xi$ as $n\to\infty$, so by a Slutsky argument we get that $\log(\left(\frac{1}{n}\right)^{1-pk}V(\mathcal{X},p,1,1;\Delta_n)_T)\overset{u.c.p.}{\to}\log(\Xi)$. By Lemma 1 above, we also know that our assumption implies that $\left(\frac{2}{n}\right)^{1-pk}V(\mathcal{X},p,1,2;\Delta_n)_T\overset{u.c.p.}{\to}\log(2\Xi)$. Hence, we get that

\begin{equation}
    \begin{split}
                \hat{H}(\mathcal{X},p,\Delta_n)_T &\overset{p}{\to} \frac{pk}{p} + \frac{\log(2\Xi) - \log(\Xi)}{\log(2)p} - \frac{1}{p} \\
        &= \frac{pk}{p}
    \end{split}
\end{equation}
\end{proof}

Having proven this result, the desired convergence follows as simple corollaries of the above. 

\begin{corollary}
    Assume that $\mathcal{X}$ is given by \eqref{eq:cont_model}, and that Assumptions \ref{assumption:kernel} and \ref{assumption:variogram} are satisfied. Then 
    \begin{equation}
        \hat{H}(\mathcal{X},p,\Delta_n)_t \overset{u.c.p}{\to}\begin{cases}
    H & \text{if }0<p<1/H\\
    H & \text{if }p>1/H
    \end{cases}
    \end{equation}
    as $n\to\infty$
\end{corollary}

\begin{proof}
    We know from \cite{Corcuera2013}, that for \eqref{eq:cont_model} and under the assumptions that 
    \begin{equation}\label{eq:cont_power_var_conv}
        \Delta_n^{1-(\alpha+1/2)p}V(\mathcal{X},p,1,1;\Delta_n)_t\overset{u.c.p}{\longrightarrow}V(\mathcal{X},p)_t:=m_p\int_0^t\vert \sigma_{1s}\vert^pds
    \end{equation}
    Hence, using Proposition \ref{prop:convergence_general} with $k=H:=\alpha+1/2$, we find that we get the claimed convergence.
\end{proof}

\begin{corollary}
    Assume that $X$ is given by \eqref{eq:jump_model} with activity index $\beta$, and that Assumptions \ref{assumption:levy_regularity}, \ref{assumption:levy_measure}, and \ref{assumption:levy_drift} are satisfied. Then
    \begin{equation}
        \hat{H}(\mathcal{X},p,\Delta_n)_t\overset{p}{\to}\begin{cases}
    1/\beta & \text{if }0<p<\beta\\
    1/p & \text{if }p>\beta
    \end{cases}
    \end{equation}
    as $n\to\infty$    
\end{corollary}

\begin{proof}
    From \cite{Todorov2010}, we know from that under our assumptions that if $p<\beta$
    \begin{equation}
        \Delta_n^{1-p / \beta} V\left(p, \mathcal{X}, \Delta_n\right)_t \stackrel{\text { u.c.p. }}{\longrightarrow} \int_0^t\left|\sigma_{2 u}\right|^p g_p\left(a_s\right) \mathrm{d} u,
    \end{equation}
    and if $p>\beta$
    \begin{equation}
        V\left(p, \mathcal{X}, \Delta_n\right)_t \stackrel{\text { u.c.p. }}{\longrightarrow} \sum_{s \leq t}\left|\Delta \mathcal{X}_s\right|^p
    \end{equation}
    From this, we see that when $p<\beta$, picking $k=1/\beta$ means that Proposition \ref{prop:convergence_general} implies that $\hat{H}(\mathcal{X},p,\Delta_n)_T$ converges to $1/\beta$. When $p>\beta$, we can pick $k=1/p$, and then the proposition implies that $\hat{H}(\mathcal{X},p,\Delta_n)_T$ converges to $1/p$.
\end{proof}

To show the convergence when we have both jumps and a continuous component we need the following lemma: 

\begin{lemma}\label{lemma:contjump_pvar_conv}
    Assume that $\mathcal{X}$ is given by \eqref{eq:cont_jump_model} and that Assumptions \ref{assumption:kernel}-\ref{assumption:levy_drift} are satisfied. Then, if $p<1/H$, we have that 
    \begin{equation}
        \Delta_n^{1-pH}V(\mathcal{X},p,1,\Delta_n)_t\overset{u.c.p.}{\to}m_p\int_0^t\vert \sigma_{1s}\vert^pds,
    \end{equation}
    and if $p>1/H$,
    \begin{equation}
        V(\mathcal{X},p,1,\Delta_n)_t\overset{u.c.p.}{\to}\sum_{s \leq t}\left|\Delta \mathcal{X}_s\right|^p.
    \end{equation}
\end{lemma}
\begin{proof}
    We assume here that $\mathcal{X}=Z=X+Y$. First, we note that we can write 
    \begin{equation}
        \begin{split}
            V(\mathcal{X},p,1,\Delta_n)_t&=\sum_{i=1}^n\vert \mathcal{X}_{i\Delta_n}-\mathcal{X}_{(i-1)\Delta_n}\vert^p\\
            &=\sum_{i=1}^n\vert X_{i\Delta_n}-X_{(i-1)\Delta_n} + Y_{i\Delta_n}-Y_{(i-1)\Delta_n}\vert^p\\
        \end{split}
    \end{equation}
    Now, if $p\leq1$, we know by the triangle inequality that 
    \begin{equation}
        \begin{split}
            V(\mathcal{X},p,1,\Delta_n)_t&\leq\sum_{i=1}^n\vert X_{i\Delta_n}-X_{(i-1)\Delta_n}\vert^p + \sum_{i=1}^nX_{i\Delta_n}-X_{(i-1)\Delta_n}\vert^p\\
            &= V(X,p,1,\Delta_n)_t + V(Y,p,1,\Delta_n)_t
        \end{split}
    \end{equation}
    and, by a similar argument 
    \begin{equation}
        \begin{split}
            V(X,p,1,\Delta_n)_t&\leq V(\mathcal{X},p,1,\Delta_n)_t + V(Y,p,1,\Delta_n)_t
        \end{split}
    \end{equation}
    which implies that 
    \begin{equation}
        \vert V(\mathcal{X},p,1,\Delta_n)_t - V(X,p,1,\Delta_n)_t\vert\leq V(Y,p,1,\Delta_n)_t.
    \end{equation}
    and hence 
    \begin{equation}\label{eq:triangle_ineq_pvar}
        \vert \Delta_n^{1-pH}V(\mathcal{X},p,1,\Delta_n)_t - \Delta_n^{1-pH}V(X,p,1,\Delta_n)_t\vert\leq \Delta_n^{1-pH}V(Y,p,1,\Delta_n)_t.
    \end{equation}    
    Now, we know that if $p<\beta$, we would need to scale $V(X,p,1,\Delta_n)_t$ by $\Delta_n^{1-p/\beta}$ for convergence, while if $p>\beta$ we would not need any scaling. Since $p<1/H$, the power on $\Delta_n^{1-p/H}$ is positive, and hence in the latter case the right hand side of \eqref{eq:triangle_ineq_pvar} converges to zero. I the former case, we know that $\beta\in(0,2)$, while $H\in(0,1/2)$, so $1-pH>1-p/\beta$, so in this case the right-hand side of \eqref{eq:triangle_ineq_pvar} also converges to zero. This means that $\Delta_n^{1-p/H}V(\mathcal{X},p,1,\Delta_n)_t$ and $\Delta_n^{1-p/H}V(X,p,1,\Delta_n)_t$ must share the same limit, which concludes the case $p\leq1$.

    In the case $p>1$, know by Minkowski's inequality that 
    \begin{equation}
            V(\mathcal{X},p,1,\Delta_n)^{1/p}_t\leq V(X,p,1,\Delta_n)^{1/p}_t + V(Y,p,1,\Delta_n)^{1/p}_t
    \end{equation}

    Now, we have two cases. First, if $p<1/H$, it follows by similar arguments as in the $p<1$ case that 
    \begin{equation}
        \begin{split}
                &\vert \Delta_n^{(1-pH)/p}V(\mathcal{X},p,1,\Delta_n)^{1/p}_t - \Delta_n^{(1-pH)/p}V(X,p,1,\Delta_n)^{1/p}_t\vert\\&\leq \Delta_n^{(1-pH)/p}V(Y,p,1,\Delta_n)^{1/p}_t.\\
        \end{split}
    \end{equation}     
    By the same arguments as earlier, $\left(\Delta_n^{1-pH}V(Y,p,1,\Delta_n)_t\right)^{1/p}\to0$, and hence $\left(\Delta_n^{1-pH}V(\mathcal{X},p,1,\Delta_n)_t\right)^{1/p}$ and $\left(\Delta_n^{1-pH}V(X,p,1,\Delta_n)_t\right)^{1/p}$ must have the same limit, and hence by the continuous mapping theorem $\Delta_n^{1-pH}V(\mathcal{X},p,1,\Delta_n)_t$ and $\Delta_n^{1-pH}V(X,p,1,\Delta_n)_t$ must also share their limit, which proves the claim for $p<1/H$.

    Now in the case $p>1/H$, using Minkowski as before, but now taking out the continuous part, we can show that 
    \begin{equation}
        \begin{split}
                &\vert V(\mathcal{X},p,1,\Delta_n)^{1/p}_t -V(Y,p,1,\Delta_n)^{1/p}_t\vert\\&\leq V(X,p,1,\Delta_n)^{1/p}_t.\\
        \end{split}
    \end{equation}     
    when $p>1/H$, we now that $1-pH<0$. This means that $\Delta_n^{1-pH}\to\infty$, and hence when we don't scale by it, $V(X,p,1,\Delta_n)^{1/p}_t$ will converge to zero. This implies, by continuous mapping theorem as before, that $V(X,p,1,\Delta_n)_t$ and $V(\mathcal{X},p,1,\Delta_n)_t$ must share their limit, which concludes the proof of the lemma. 
\end{proof}

With this lemma in place, we are ready to prove the following corollary: 

\begin{corollary}
    Assume that $\mathcal{X}$ is given by \eqref{eq:cont_jump_model}, and that Assumptions \ref{assumption:kernel}-\ref{assumption:levy_drift} are satisfied. Then
    \begin{equation}
        \hat{H}(\mathcal{X},p,\Delta_n)_t\overset{u.c.p.}{\to}\begin{cases}
    H & \text{if }0<p<1/H\\
    1/p & \text{if }p>1/H
    \end{cases}
    \end{equation}
    as $n\to\infty$    
\end{corollary}

\begin{proof}
    We will again use Proposition \ref{prop:convergence_general} together with Lemma \ref{lemma:contjump_pvar_conv}. When $p<1/H$, we know that that picking $k=H$ will lead to the the required convergence of the power variation. When $p>1/H$, picking $k=1/p$ will kill the scaling, which is what we need for convergence. 
\end{proof}

\begin{lemma}
    The convergence of the roughness signature function is uniform in $p$. 
\end{lemma}
\begin{proof}
    This proof is practically the same as the one in \cite{Todorov2010}. First, we know that the roughness signature function is a continuous transformation of $V(\mathcal{X},p,1,\Delta_n)_T$ and $V(\mathcal{X},p,2,\Delta_n)_T$, we are done if we can show that 
    \begin{equation}
        \Pi_n(p):=\left(\begin{array}{c}
        \Delta_n^{1-p / k} V\left(p, \mathcal{X}, 2, \Delta_n\right)_T \\
        \Delta_n^{1-p / k} V\left(p, \mathcal{X}, 1,\Delta_n\right)_T
        \end{array}\right) \stackrel{p}{\longrightarrow}\left(\begin{array}{c}
        2^{p / k} \Xi_T(p) \\
        \Xi_T(p)
        \end{array}\right) \text { uniformly on }\left[p_l, p_u\right]        
    \end{equation}
    where $k$ is chosen such that it fits the underlying model (i.e. $k=1/\beta$ for the continuous model with $p_u<\beta$, $k=H$ for the continuous or continuous plus jumps model with $p<1/H$. The pure-jump model with $p>\beta$ and continuous model with $p>1/H$ can be handeled with the same arguments but without the $\Delta_n$ scaling.). For this, we know that we need finite-dimensional convergence and tightness (\cite{limitTheorems}, Proposition IV, 1.17 (b)). We have already established the finite dimensional convergence, so we just need to prove tightness. To do this, we introduce the modulus of continuity \citep{limitTheorems}:
    \begin{equation}
        w\left(\Pi_n, \theta\right):=\sup \left\{\sup _{u, v \in[p, p+\theta]}\left|\Pi_n(u)-\Pi_n(v)\right|: p_l \leq p \leq p+\theta \leq p_u\right\}
    \end{equation}
    Now, we note the the power variation if of the form $x\mapsto \sum x^p$ for some positive values $x$, which means that the terms in the sum are either maximized when $p$ is as large as possible or as small as possible, depending on whether $x$ is larger or smaller than $1$. This means that for a sufficiently small $\theta$ we get the inequality 
    \begin{equation}
        \begin{aligned}
            w\left(\Pi_n, \theta\right) \leq & U\left(\mathcal{X}, p_l+\theta, p_l, \Delta_n\right)_T+U\left(\mathcal{X}, p_u, p_u-\theta, \Delta_n\right)_T \\
            & +U\left(\mathcal{X}, p_l+\theta, p_l, k \Delta_n\right)_T+U\left(\mathcal{X}, p_u, p_u-\theta, k \Delta_n\right)_T
        \end{aligned}
    \end{equation}
    where
    \begin{equation}
        U\left(\mathcal{X}, p, q, \Delta_n\right)_T=\Delta_n \sum_{i=1}^{\left[T / \Delta_n\right]}|| \Delta_n^{-1 / \beta} \Delta_i^n \mathcal{X}\vert^p-|\Delta_n^{-1 / \beta} \Delta_i^n \mathcal{X}|^q \vert
    \end{equation}
    
    Next, so for $q$ in a neighbourhood $(p-\theta,p+\theta)$ around $p$, Taylor's theorem with the mean value theorem implies that 
    \begin{equation}
        \begin{split}
            \sum x^q &= \sum x^p + \frac{\partial}{\partial p}\left(\sum x^p\right)(p)\vert p-q\vert + \xi_2(q)\\
            &= \sum x^p + \sum\log(x)x^p\vert p-q\vert + \xi_2(q)\\
            &\leq \sum x^p + K\sum x^p \vert p-q\vert + \xi_2(q)\\
            \Rightarrow \sum x^p - \sum x^q &\leq K\sum x^\gamma \vert p-q\vert\\
            \Rightarrow \sum x^p - \sum x^q &\leq K\left(\sum x^{p\wedge q-\varepsilon} + \sum x^{p\vee q+\varepsilon}\right) \vert p-q\vert
        \end{split}
    \end{equation}
    where $\gamma\in(p-\theta,p+\theta)$. This implies the inequality 

    \begin{equation}
        U(\mathcal{X},p,q,\Delta_n)_t\leq K\vert p-q\vert\Delta_n\sum_{i=1}^{[t/\Delta_n]}\left(\vert\Delta_n^{-1/\beta}\Delta_i^n \mathcal{X}\vert^{p\wedge q-\varepsilon} + \vert \Delta_n^{-1/\beta}\Delta_i^n\mathcal{X}\vert^{p\vee q+\varepsilon}\right)
    \end{equation}
    for an arbitrarily small $\varepsilon$ and some constant $K$. Now by appealing to Ascoli-Arzela's criteria for tightness (Proposition IV.3.26 (ii) of \cite{limitTheorems}), we can establish the tightness. Condition IV.3.21 is clearly satisfied, as we are looking over compact sets, while IV.3.27 follows by the inequality above. 
    
    
\end{proof}

\end{section}

\end{document}